\documentclass{article}
\usepackage{arxiv}
\usepackage[utf8]{inputenc}
\usepackage{times}

\usepackage{soul}
\usepackage{url}
\usepackage[utf8]{inputenc}
\usepackage[small]{caption}
\usepackage{graphicx}
\usepackage{amsmath}
\usepackage{amsthm}
\usepackage{thmtools}
\usepackage{thm-restate}
\usepackage{booktabs}
\usepackage{algorithm}
\usepackage[noend]{algorithmic}
\usepackage{mathtools}
\usepackage[dvipsnames]{xcolor}

\title{On verifying expectations and observations of intelligent agents}

\author{
Sourav Chakraborty \\
Indian Statistical Institute, Kolkata, India\\
sourav@isical.ac.in\\
\And
Avijeet Ghosh\\
Indian Statistical Institute, Kolkata, India\\
avijeet\_r@isical.ac.in\\
\And
Sujata Ghosh \\
Indian Statistical Institute, Chennai, India\\
sujata@isichennai.res.in\\
\And
Fran{\c{c}}ois Schwarzentruber\\
ENS Rennes\\
francois.schwarzentruber@ens-rennes.fr\\
}
\usepackage{amsmath,amssymb, color, tikz}
\usetikzlibrary{positioning}%added by Avijeet
\tikzstyle{circNode} = [circle, text centered, draw=black, minimum size= 2em]

\newcommand{\LL}{\mathcal{L}} %for automata
\newcommand{\BP}{\ensuremath{\mathbf{P}}}
\newcommand{\Act}{\ensuremath{\mathbf{\Sigma}}}
\newcommand{\Ag}{\ensuremath{\mathbf{I}}}

\newcommand{\M}{\mathcal{M}}

\newcommand{\tr}{\mathsf{tr}}

\newcommand{\regdiv}[1]{\ensuremath{\backslash} #1}

\newcommand{\ep}{\ensuremath{\varepsilon}}
\newcommand{\dl}{\ensuremath{\delta}}
\newcommand\ldiaarg[1]{\langle#1\rangle}
\tikzstyle{world} = [draw]
\renewcommand{\phi}{\varphi}

\newcommand{\Exp}{\it Exp} %math mode is always \it

\newcommand{\imp}{\rightarrow}

\newcommand{\POL}{\mathsf{POL}}
\newcommand{\EL}{\mathsf{EL}}
\newcommand{\DEL}{\mathsf{DEL}}
\newcommand{\PDL}{\mbox{\rm PDL}}
\newcommand{\EPL}{\mbox{\rm EPL}}

\newcommand{\PAL}{\mathsf{PAL}}
\newcommand{\APAL}{\mathsf{APAL}}

\newcommand{\union}{\cup}

\newcommand\algoaccept{\textbf{accept}}
\newcommand\algoreject{\textbf{reject}}

\newcommand{\ldiamondarg}[1]{\langle#1 \rangle}

\newcommand{\mcNP}{\mathsf{mcSF\&EXIT}}
\newcommand{\mcWORDS}{\mathsf{mcWords}}
\newcommand{\DecidePSPACE}{\mathsf{mcPOL}}

\newcommand{\ResidueByLetter}{\mathsf{ResidueByLetter}}

\newcommand{\AuxOut}{\mathsf{AuxOut}}
\newcommand{\GetSetNP}{\mathsf{WorldsNP}}
\newcommand{\GetSet}{\mathsf{WorldsP}}
\newcommand{\T}{\mathsf{True}}
\newcommand{\Fa}{\mathsf{False}}
\newcommand{\starfree}{\mathsf{Star\mbox{-}Free}}
\newcommand{\word}{\mathsf{Word}}
\newcommand{\existential}{\mathsf{Existential}}
\newcommand{\PSPACE}{\mathsf{PSPACE}}
\newcommand{\PTime}{\mathsf{P}}
\newcommand{\NP}{\mathsf{NP}}
\newcommand{\PTIME}{\mathsf{PTIME}}
\newcommand{\automaton}{\mathcal A}
\newcommand{\modelM}{\mathcal M}

\newcommand{\set}[1]{\{#1\}}
\newcommand{\suchthat}{\mid}

\newcommand{\obsright}{\blacktriangleright}
\newcommand{\obsleft}{\blacktriangleleft}
\newcommand{\obsup}{\blacktriangle}
\newcommand{\obsdown}{\blacktriangledown}

\newcommand{\expwater}{(\obsright \union \obsup)^* (\obsdown \union \obsleft \union \ep) (\obsright \union \obsup)^*}
\newcommand{\exppower}{(\obsleft \union \obsdown)^* (\obsup \union \obsright \union \ep) (\obsleft \union \obsdown)^*}
\newcommand{\exppatrol}{(\obsright^+ \obsdown^+ \obsleft^+ \obsup^+)^*}

\newcommand{\prefixes}{\mathit{prefixes}}

\newtheorem{theorem}{Theorem}

\newtheorem{lemma}[theorem]{Lemma}

\newtheorem{definition}{Definition}

\newtheorem{example}{Example}

\newcommand{\emptyword}{\epsilon}

\begin{document}

\maketitle

\begin{abstract}
Public observation logic ($\POL$) is a variant of dynamic epistemic  logic  to  reason  about  agent  expectations  and  agent observations. Agents have certain expectations, regarding the situation at hand, that are actuated by the relevant protocols, and they eliminate possible worlds in which  their expectations do not  match with  their  observations.  In this work, we  investigate  the  computational complexity of the model checking problem for $\POL$ and  prove its $\PSPACE$-completeness. We also study various syntactic fragments of $\POL$. We exemplify the applicability of $\POL$ model checking in verifying different characteristics and features of an interactive system with respect to the distinct expectations and (matching) observations of the system. Finally, we provide a discussion on the implementation of the model checking algorithms.
\end{abstract}

\section{Introduction}\label{intro}
%%context
Agents have expectations about the world around, 
and they reason on the basis of what they observe around them, and such observations may or may not match with the expectations they have about their surroundings. Let us first provide two examples showing the diverse nature of such reasoning phenomena. 
\begin{itemize}
	\item Consider a person traveling from  Switzerland to France in a car. Here is one way she would know whether she is in France. According to her expectations based on the traffic light signals of the different states, if she observes the sequence of (green$^\ast$-amber-red$^\ast$)$^\ast$ ($\ast$ denotes the continuance of such sequences), she would know that she is in France, whereas if she observes (green$^\ast$-amber-red$^\ast$-amber)$^\ast,$ she would know that she is not.
	\item Consider three agents denoted by Sender (S), Receiver (R) and Attacker (A). Suppose S and R have already agreed that if S wants to convey that some decision has been taken, S would send a message, say $m$, to R; otherwise, S would send some other message, say $m'$, to R. Suppose also that A has no information about this agreement. Then upon getting a message from S, there would be a change in the knowledge state of R but not A.
\end{itemize}
The first example concerns a certain rule that we follow in our daily life, and the second example brings in the flavour of coded message-passing under adversarial attacks. Expectations about the moves and strategies of other players also occur naturally in game theory, and possible behaviours of players are represented in these terms.  Moving from theory to actual games, in the strategy video game Starcraft\footnote{\url{https://en.wikipedia.org/wiki/StarCraft_(video_game)}}, one player may know/expect that the other player will attack her base as soon as possible, and thus may play accordingly. Games like Hanabi\footnote{\url{https://en.wikipedia.org/wiki/Hanabi_(card_game)}}, and Colored Trails \cite{DBLP:journals/aamas/WeerdVV17} also consider the connection between expectations and observations regarding the moves and strategies of the other players.

%%research question
The challenge now is to build intelligent systems that are able to reason about knowledge regarding expectations, and plan accordingly. 
%lack of the literature to solve the research question
Whereas epistemic logic \cite{RAK} and more generally, its dynamic extensions, popularly known as \emph{dynamic epistemic logics} ($\DEL$) \cite{DEL} help us to build agents that reason about knowledge, they do not offer any mechanism dealing with expectations. 
In the same way, epistemic planning, based on the model checking of $\DEL$
(\cite{DBLP:journals/ai/BolanderCPS20}), 
extends classical planning with epistemic reasoning,
but is unable to take agent expectations into account.
Fortunately, following \cite{DBLP:conf/icla/Wang11}, 
\emph{Public observation logic} ($\POL$)  \cite{van2014hidden}, a variant of $\DEL$, reasons about knowledge regarding expectations. $\POL$ provides dynamic operators for verifying whether a given epistemic property holds after observing some sequence of observations matching certain expectations that are modelled by regular expressions~$\pi$.

%one possible answer to the research question!
However, investigations on algorithmic properties of $\POL$ were left open. In this paper, we show that the $\POL$ model checking is decidable and $\PSPACE$-complete. Our result relies on automata theory and the careful use of an oracle for deciding the algorithm running in poly-space.
%language residues in the context of epistemic logic.

%
For practical purposes, we investigate syntactic fragments that offer better complexities than reasoning in the full language of $\POL$ (see Figure~\ref{figure:results}), and are suitable for relevant verification tasks:
\begin{itemize}
    \item the $\word$ fragment, where any regular expression $\pi$ is a \textit{word}, is sufficient to verify that some given plan leads to a state satisfying some epistemic property;
\item the $\existential$ fragment, where the dynamic operators of $\POL$ are all \emph{existential}, is suitable for epistemic planning (e.g., does there exist a plan?);
\item the $\starfree$ fragment, where the regular expressions $\pi$ are \textit{star-free}, embeds \emph{bounded} planning (in which sequences of observations to synthesize are bounded by some constant). In particular, the $\starfree\existential$ fragment (i.e. the intersection of the $\starfree$ and the $\existential$ fragments) is suitable for bounded epistemic planning.
\end{itemize}

\begin{figure}
	\begin{center}
		\newcommand{\ybottom}{-1.2}
		\newcommand{\sep}[1]{\draw (#1, 1.6) -- (#1, \ybottom);}
		\tikzstyle{box} = [inner sep=0.4mm]
		\tikzstyle{inclusion} = [-latex, line width=0.5mm]
		\scalebox{0.9}{
			\begin{tikzpicture}[xscale=2.2, yscale=1]
				\draw (-1.4, \ybottom) rectangle (2.8, 1.6);
				\draw (-1.4, 1) -- (2.8, 1);
				\node[text width=3cm,box] (full) at (2.5, 0) {Full language (Th.~\ref{theorem:full})};
				\node[text width=15mm,box] (starfree) at (1.1, 0.5) {$\starfree$ (Th.~\ref{theorem:starfree})};
				\node[text width=14mm] (exist) at (1.1, -0.5) {$\existential$ (Th.~\ref{theorem:existential})};
				\node[text width=15mm] (starfreeexist) at (0, -0.5) {$\starfree$ $\existential$ (Th.~\ref{theorem:starfreeexist})};
				\node[text width=8mm] (word) at (-1, 0.5) {$\word$ (Th.~\ref{theorem:word})};
				\node at (1.5, 1.2) {$\PSPACE$-complete};
				\node at (0, 1.2) {$\NP$-complete};
				\node at (-1, 1.24) {in $\PTIME$};
				\sep {0.55}
				\sep {-0.6}
				\draw[inclusion] (starfree) -- (full);
				\draw[inclusion] (exist) -- (full);
				\draw[inclusion] (starfreeexist) edge[out=10, in=200] (starfree);
				\draw[inclusion] (starfreeexist) -- (exist);
				\draw[inclusion] (word) -- (starfree);
		\end{tikzpicture}}
	\end{center}
	\caption{Complexity results of model checking for different fragment of $\POL$. (arrows represent inclusion of fragments).\label{figure:results}}
\end{figure}
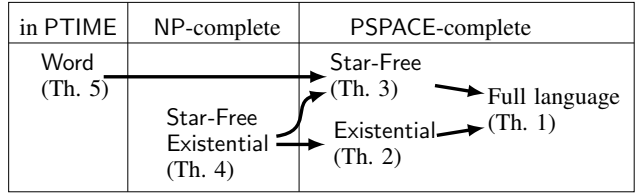

\paragraph{Outline.} Section 2 recalls $\POL$ with a formal presentation of the two examples mentioned in the introduction. Section 3 deals with all our complexity results about the model checking problem of $\POL$. Section 4 shows the applicability of $\POL$ and its fragments in modelling interactive systems. It also includes a discussion on the implementation. Section 5  presents the related work and section 6 concludes the paper. %the computational complexity of model checking problems of several fragments of POL and section 5 concludes the paper.\todo{rewrite, tsections do not have numbers!}

\section{Background and Preliminaries}\label{back}

We first provide an %a brief 
overview of public observation logic ($\POL$) as introduced in \cite{van2014hidden}. Let $\Ag$ be a finite set of agents,  $\BP$ be a countable
%François: finite is not enough for the complexity results
set of propositions describing the facts about the world and $\Act$ be a finite set of actions. Below, we will not differentiate between the action of observing a phenomenon and the phenomenon itself.

\subsection{Observations}\label{obs} 

For our purposes, we assume \textit{observations} to be finite strings of actions. In the traffic example, an observation may be \emph{green-amber-red-green} (abbreviated as $garg$) or, \emph{green-amber-red-amber-green} (abbreviated as $garag$), among others, whereas, in the message-passing example, an observation is either $m$ or $m'$.  An agent may expect different (even infinitely many) potential observations %to happen 
at a given state, but to model %human/
agent expectations, %such expectations 
they are described in a finitary way by introducing the {\em observation expressions} (as regular expressions over $\Act$): 

% \begin{definition}[Observation expressions]
% 	Given a finite set of action symbols $\Act$, the language $\mathcal{L}_{\it obs}$ of {\em observation expressions} is defined by the following BNF:
% 	$$\begin{array}{r@{\quad::= \quad}l}
% 		\pi  &
% 		\emptyset\mid\
% 		\ep\mid
% 		a
% 		\mid \pi\cdot \pi
% 		\mid \pi + \pi
% 		\mid \pi^*\\
% 	\end{array}$$ 
% 	\noindent where $\emptyset$ denotes the empty set of observations,  the constant $\ep$ represents the empty string, and $a\in\Act$.
% 	\label{definition:obsexpression}
% 	%We call the set of all these expressions $\Pi.$
% \end{definition}

\begin{definition}[Observation expressions]
	Given a finite set of action symbols $\Act$, the language $\mathcal{L}_{\it obs}$ of {\em observation expressions} is defined by the following BNF:
	$$\begin{array}{r@{\quad::= \quad}l}
		\pi  &
		\emptyset\mid\
		\ep\mid
		a
		\mid \pi\cdot \pi
		\mid \pi + \pi
		\mid \pi^*\\
	\end{array}$$ 
	\noindent where $\emptyset$ denotes the empty set of observations,  the constant $\ep$ represents the empty string, $\cdot$ denotes concatenation, $+$ is union,  $*$ represents iteration and $a\in\Act$.
	\label{definition:obsexpression}
	%We call the set of all these expressions $\Pi.$
\end{definition}

In the traffic example, the observation expression $(g^\ast ar^\ast)^\ast$ models the traveller's expectation of traffic signals in case she is in France. In the other one, the expression $m$ models the expectation of the receiver in case a decision is made.

\subsection{Models}	

	{\em Epistemic expectation models} %from
	\cite{van2014hidden} capture the expected observations of agents. 
	They can be seen as epistemic models \cite{RAK} together with, for each world, a set of potential or expected observations.
	Recall that an epistemic model is a tuple $\langle S, \sim ,V\rangle$ where $S$ is a non-empty set of worlds, $\sim$ assigns to each agent in $\Ag$ an equivalence relation $\sim_i \subseteq S \times S$, and $V : S \rightarrow 2^{\BP}$ is a valuation function.
	
	\begin{definition}[Epistemic expectation model]
		\label{Model1}
		An {\em epistemic expectation model} $\M$ is a quadruple $\langle S, \sim ,V, \Exp\rangle,$ where $\langle S, \sim ,V\rangle$ is an epistemic model %(the {\em epistemic skeleton} of $\M$) 
		and  $\Exp : S \imp \LL_{\it obs}$ is an expected observation function assigning to each world an observation expression $\pi$ such that $\LL(\pi)\not=\emptyset$ (non-empty set of finite sequences of observations).   %expected to be `observed' at that state). 
		%An {\em epistemic expectation state} is a pointed epistemic expectation model $\langle S, \sim ,V, \Exp,s\rangle$. 
		A pointed epistemic expectation model is a pair $(\M, s)$ where $\M = \langle S, \sim ,V, \Exp\rangle$ is an epistemic expectation model and $s \in S$.
	\end{definition}  
	
	% \begin{example}
		%     \begin{tikzpicture}[node distance=2cm]
			
			%         \node(s) [circNode] {$s$};
			%         \node(f) [circNode, right of=s] {$f$};

			%         \draw[-] (s) -- node[anchor=north] {sender} (f);
			
			%     \end{tikzpicture}
		%     An Amazon seller when sends a package out for delivery, it waits and expects for a "success" or a "failure" message, and accordingly \textit{knows} whether the package was delivered. This scenario can be modeled using $\M = \langle S,V,\{R_{sender}\},Exp\rangle$ the worlds $S = \{s, f \}$ and the relation $R_{sender} = \{(s,s), (f,f), (s,f), (f,s) \}$ as shown in the figure (the reflexive adjascency is not shown). Let $p$ be a propositional letter denoting \textit{"The package was delivered"}. Hence $V(s) = \{p\}$ and $V(f) = \{\neg{p}\}$. Also we denote the expectations of the agent sender as $Exp(s) = a$ and $Exp(f) = b$, where $a$ denote the letter representing the success reply, and $b$ as the failure reply.
		% \end{example}
	
	Intuitively, $\Exp$ assigns to each world a set of potential or expected observations. We now provide the model definitions of the examples mentioned in the introduction. The traffic light example where only one agent (the traveller) is involved can be depicted by the model $\M_{tl}$ (cf. Figure \ref{figure:traffic}). Unless the traveller ($T$) observes the respective sequences of traffic signals, she would not know whether she is in France ($f$) or not ($\neg f$). Her uncertainty is represented by the (bi-directional) link between the two worlds $s$ and $t$. 
	For the sake of brevity, 
	%we have not added 
	we do not draw
	the reflexive arrows. 
	Similar representations are used in the message-passing example as well (cf. Figure~\ref{figure:message}). Here, the receiver would get to know about the decision depending on the message he receives, whereas, the attacker would be ignorant of the fact irrespective of the message ($m$ or $m'$) she receives.
	\par
	Given an {\em epistemic expectation model} $\M = \langle S, \sim ,V, \Exp\rangle$, the size of these observation expressions is defined as follows.

\begin{definition}[Size of observation expressions]
	    Given an observation expression $\pi$, the size of $\pi$, $|\pi|$ can be defined inductively as follows:
	    
        $\begin{array}{l}
            |\ep| = |\emptyset| = 0\\
        	|a| = 1\\
        	|\pi\cdot \pi'|= |\pi + \pi'|=|\pi| + |\pi'| + 1\\
        	|\pi^*|= |\pi| + 1
        \end{array}$
        
\end{definition}

In the traffic example, the observation expression $(g^\ast ar^\ast)^\ast$ models the traveller's expectation of traffic signals in case she is in France.
Evidently, in this case, the size of the observation expression, $(g^\ast ar^\ast)^\ast$, is 6. The semantics for the observation expressions 
are given by {\em sets of  observations} (strings over $\Act$), similar to those for regular expressions \cite{DBLP:books/daglib/0086373}.

\begin{definition}[Semantics of observation expressions]
Given an observation expression $\pi$, the corresponding {\em set of observations}, denoted by $\LL(\pi)$, is the set of finite strings over $\Act$ defined as follows:

$\begin{array}{l}
	\LL(\emptyset)=\emptyset\\
	\LL(\ep)=\{\epsilon\}\\
	\LL(a)=\{a\}\\
	\LL(\pi\cdot \pi')= \{wv \mid w \in \LL(\pi) \mbox{ and } v \in \LL(\pi')\}\\
	\LL({\pi} + \pi')=\LL(\pi)\cup\LL(\pi')\\
	\LL(\pi^*)=\{\ep\}\cup\bigcup_{n> 0}(\LL(\underbrace{\pi\cdots\pi}_n))
	\end{array}$
	
\end{definition}

For example, $\LL(m)=\{m\}$, and $\LL((g^\ast ar^\ast)^\ast) = \{\ep, a, ga, ar, gar, gargar, \ldots \}$. Before going any further, let us first introduce the notion of $\emph{residue}$ which play an important role in providing the semantics of $\POL$:

\begin{definition}[Residue]\label{def:residue}
Given an observation expression $\pi$ and a word (finite string over $\Act$) $w$, a residue of $\pi$ with respect to $w$, denoted by $\pi\regdiv w$, is an observation expression defined with an auxiliary output function $o$ from the set of observation expressions over $\Act$ to $\{\emptyset, \ep\}$. If $\ep\in \LL(\pi)$, then $o$ maps $\pi$ to \ep; otherwise,it maps $\pi$ to $\emptyset$:  
	%\[
    %o(\pi)= 
    %\begin{cases}
     %   \epsilon,& \text{if } \epsilon\in \LL(\pi)\\
      %  \delta,              & \text{otherwise}
    %\end{cases}
 %   \]
%	:\\
% 	\todo{the reader may need a sentence to present the following equations. I guess it is a recursive definition of $o$ etc. but I am not sure. Also explain $\pi\regdiv w$ in words}
% 	\todo{maybe all of this will be in the supplementary materials. It is just to explain better but these things are not from us!}
%	\medskip
%	\begin{minipage}[b]{0.5\linewidth}
	%\[
	
	%$\begin{array}{l@{\quad}l}
	%	\pi=o(\pi)+\sum_{a\in\Act}(a\cdot \pi\regdiv{a}) \\
	%	o(\ep )=\ep \\
	%	o(\dl )=o(a)=\dl\\
	%	o(\pi+\pi')=o(\pi)+ o(\pi')\\
	%	o(\pi\cdot\pi)=o(\pi)\cdot o(\pi')\\
	%	o(\pi^*)=\ep\\
	%	\end{array}
	%$
%	\end{minipage}
%	\medskip
	% $\;$\\
%	\begin{minipage}[b]{0.5\linewidth}
	$\begin{array}{l@{\quad}l}
		\ep  \regdiv{a}=\dl \regdiv{a}= b\regdiv{a}=\dl  \quad(a\not= b)\\
		a \regdiv{a} =\ep \\
		(\pi+\pi') \regdiv{a}=\pi\regdiv{a} + \pi'\regdiv{a}\\
		(\pi\cdot\pi') \regdiv{a}=(\pi\regdiv{a})\cdot\pi'+o(\pi)\cdot(\pi'\regdiv{a})\\
		\pi^*\regdiv{a}=\pi\regdiv{a}\cdot \pi^*\\
		\pi\regdiv{a_0\cdots a_n}=\pi\regdiv{a_0}\regdiv{a_1}\dots\regdiv{a_n}\\
		\end{array}
		$
%		\end{minipage}
% 		\todo{give the definition and finish with $\pi=o(\pi)+\sum_{a\in\Act}(a\cdot \pi\regdiv{a})$ which is NOT a definition but a proposition, no?}
		\medskip
		
		\end{definition}
		%
		%
		%\noindent
		%The above construction of the output function helps to compute the residual of compositions. Reading from left to right the above equations can be viewed as rewriting rules which push the $\regdiv{a}$ operation to the `inner' part of the expression and finally eliminate them. Thus by using these equations residuals of observations can be computed syntactically. \\ 
Intuitively, the regular language $\pi\regdiv w$ is the set of words given by  $\{v \in \Act^* \mid wv\in\LL(\pi)\}$. The regular language $\prefixes(\pi)$ is the set of prefixes of words in $\LL(\pi)$, that is, $w\in \prefixes(\pi)$ 
			iff $\exists v\in \Sigma^*$ such that $wv\in\LL(\pi)$ 
			iff $\exists v\in \Act^*$ such that $wv\in\LL(\pi)$ (namely, $\LL(\pi\regdiv w)\not=\emptyset$)
\begin{example}
$(g^*ar^*a)^*\regdiv (garaga) = r^*(g^*ar^*a)^*$ denotes the language of words $\{v: garaga\cdot v \in \LL((g^*ar^*)^*)\}$. The set $\prefixes((g^*ar^*a)^*)$ contains $garaga$. However, $garg$ is not in $\prefixes((g^*ar^*a)^*)$ and $(g^*ar^*a)^*\regdiv (garg)$ is empty.
\end{example}

We recall that by Thomson's construction \cite{DBLP:books/daglib/0086373}, there is a non-deterministic finite automaton (NFA) of polynomial size in $|\pi|$, that recognizes the language  $\LL(\pi)$. For simplicity in notations, that NFA is also denoted by $\pi$.

Given a word $w$, $\pi\regdiv w$ denotes the regular language $\{v\mid wv\in\LL(\pi)\}$. $\pi\regdiv w$ corresponds to right residuation with respect to the monoid $(\Act^*, \cdot, \ep)$. The representation of $\pi\regdiv w$ is the NFA $\pi$ augmented by the subset of states of $\pi$ that one reaches after having read $w$ from the initial states of $\pi$. Computing the representation of $\pi \regdiv w$ is polynomial in $|\pi|$ and $|w|$.
	In the sequel, $DFA(\pi)$ denotes the minimal deterministic finite automaton for $\pi$ (unique up to isomorphism). Computing $DFA(\pi)$ is exponential in $|\pi|$ in the worst-case.

	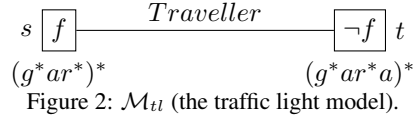
\begin{figure}
		\begin{center}
			\begin{tikzpicture}
				\node[world] (s) {$f$};
				\node[world] (t) at (4, 0) {$\lnot f$};
				\node[left = 0mm of s] {$s$};
				\node[right = 0mm of t] {$t$};
				\node[below = 0mm of s] {$(g^*ar^*)^*$};
				\node[below = 0mm of t] {$(g^*ar^*a)^*$};
				\draw (s) edge node[above] {$Traveller$} (t);
			\end{tikzpicture}
		\end{center}
		\vspace{-.5cm}
		\caption{$\M_{tl}$ (the traffic light model)\label{figure:traffic}.}
	\end{figure}
	
	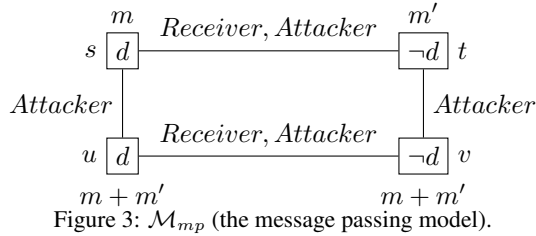
\begin{figure}
		\begin{center}
			\begin{tikzpicture}[yscale=1.4]
				\node[world] (s) {$d$};
				\node[world] (t) at (4, 0) {$\lnot d$};
				\node[world] (u) at (0, -1) {$d$};
				\node[world] (v) at (4, -1) {$\lnot d$};
				\node[left = 0mm of s] {$s$};
				\node[right = 0mm of t] {$t$};
				\node[left = 0mm of u] {$u$};
				\node[right = 0mm of v] {$v$};
				\node[above = 0mm of s] {$m$};
				\node[above = 0mm of t] {$m'$};
				\node[below = 0mm of u] {$m+m'$};
				\node[below = 0mm of v] {$m+m'$};
				\draw (s) edge node[above] {$Receiver, Attacker$} (t);
				\draw (u) edge node[above] {$Receiver, Attacker$} (v);
				\draw (s) edge node[left] {$Attacker$} (u);
				\draw (t) edge node[right] {$Attacker$} (v);
			\end{tikzpicture}
		\end{center}
		\vspace{-.5cm}
		\caption{$\M_{mp}$ (the message passing model)\label{figure:message}.}
	\end{figure}

	The main idea for introducing this logic was to reason about agent knowledge via the matching of observations and expectations. In line of public announcement logic \cite{DBLP:journals/synthese/Plaza07}, it is assumed that when a certain phenomenon is observed, people delete some impossible scenarios where they would not expect that observation to happen.
	To this end, the update of epistemic expectation models according to some observation $w\in\Act^*$ is defined below. The idea behind an updated expectation model  is %that we delete 
	to delete the worlds where the  observation $w$ could not have been happened.
	
	\begin{definition}[Update by observation]\label{def.upobs}
		Let $w$ be an observation over $\Act$ and let $\M=\langle S,\sim,V,\Exp\rangle$ be an epistemic expectation model. The updated model $\M|_w= \langle S',\sim',V',\Exp'\rangle$ is defined by: $S' = \{s \in S \mid \LL(\Exp(s)\regdiv w)\not=\emptyset\}$, ${\sim'_i}={\sim_i}|_{S'\times
			%\Ag\times <= François: wrong
			S'},$ $V'=V|_{S'},$ and for all $s\in S'$, $\Exp'(s)=\Exp(s)\regdiv w$.
	\end{definition}
	
	%As mentioned above, the updated set of worlds $S'\subseteq S$ constitute those worlds $s$ where the word $w$ can be observed, which is taken care of by the condition $\LL(\Exp(s)\regdiv w)\not=\emptyset$. The definitions of $\sim'$ and $V'$ are given by usual restrictions to $S'$. Finally, the expectation at each of the worlds in $S'$ gets updated to the regular expression corresponding to the set of observations (finite strings of actions) that remain after the removal of $w$ from these strings. 
In Definition~\ref{def.upobs}, $S'$ is the set of worlds $s$ in $S$ where the word $w$ can be observed, i.e., $\LL(\Exp(s)\regdiv w)\not=\emptyset$. The definitions of $\sim'$ and $V'$ are given by usual restrictions to $S'$. The expectation at each world in $S'$ gets updated by observing the word $w$: finite strings of actions that are of the form $wu$ are replaced by $u$ while strings that are not of the form $wu$ get removed because they do not match the expectation.

\begin{example}
    Consider the model $\M_{tl}$ of Figure~\ref{figure:traffic} and  $w = garga$. The updated model $\M_{tl}|_w =\langle S',\sim',V',\Exp'\rangle$ is such that $S' = \{s\}$: world $t$ is removed because $garga$ is not a prefix of any word in $\LL((g^*ar^*a)^*)$. The expectation $\Exp(s)$ is replaced by $\Exp'(s) = (g^*ar^*)^*\regdiv (garga) = r^*(g^*ar^*)^*$.
\end{example}
	
	%
% 	\todo{François things that this can be cut off}

		%The idea of the updated model is %that we delete those
		%to delete the worlds where the observation $w$ could not have happened.
		
		\subsection{Public observation logic ($\POL$)}
		
		To reason about agent expectations and observations, the language for $\POL$ is provided below.
		
		\begin{definition}[Syntax]  %Given are a countable set of {\em propositional variables} $\BP$, a finite sets of {\em actions} $\Act$, and a finite set of {\em agents} $\Ag$, 
			The {\em formulas} $\phi$ of $\POL$ 
			are given by:
			
			\vspace{.1cm}
			
		     $\begin{array}{r@{\quad::= \quad}l}
				\phi  &
				\top
				\mid
				p
				\mid \neg \phi
				\mid \phi \land \phi
				\mid K_i\phi
				\mid [\pi] \phi,
				% 	   \mid [!\pi]\phi
				% \pi  &
				%\ep\mid
				% \dl\mid	
				%a
				% 	  \mid ?\phi_b
				%\mid \pi\cdot \pi
				% \mid \pi + \pi
				% \mid \pi^*\\
			\end{array}$
			
			\vspace{.1cm}
			
			\noindent where $p\in\BP$, $i\in\Ag$, and $\pi\in\mathcal{L}_{\it obs}$. 
		\end{definition}
% 		Intuitively, $K_i\phi$ says that `agent $i$ knows that $\phi$', and $[\pi]\phi$ says that `after any observation in $\pi$, $\phi$ holds'. The other propositional connectives are defined in the usual manner. We also define $\ldiaarg{\pi}\phi$ as $\lnot [\pi]\lnot\phi$ and $\hat{K}_i\phi$ as $\lnot K_i \lnot\phi$. We will mostly use these modalities in our proofs. The \textbf{$\starfree$ fragment of $\POL$} is the set of formulas in which the $\pi$'s do not contain any Kleene star $*$. A much more restricted version is the \textbf{$\word$ fragment of $\POL$}, where $\pi$'s are words. The \textbf{$\existential$ fragment of $\POL$} is the set of formulas for which there is an odd number of negations in front of $K_i$ and $[\pi]$ modalities. Equivalently, it corresponds to formulas in negative normal form in which only the operators $\ldiaarg{\pi}$ and $\hat{K}_i$ appear. Finally, we have the \textbf{$\starfree$ $\existential$ fragment of $\POL$} which is the $\existential$ fragment with the extra guarantee that the $\pi$'s do not contain any Kleene star $*$.

% \section{Various Fragments of $\POL$}
Also further we talk about various fragments, or special cases, of Public Observation Logic. Here various restrictions are promised on the syntax of regular languages or the syntax of the formulas. But before that, we define a normal form of $\POL$ formulas called the Negative Normal Form.
\begin{definition}[Negative Normal Form (NNF)]
Given a set of propositional letters $\BP$, the Negative Normal Form of a $\POL$ formula is defined recursively as following:
$$\begin{array}{r@{\quad::= \quad}l}
				\phi  &
				\top
				\mid
				p
				\mid \neg p
				\mid \phi \land \phi
				\mid K_i\phi
				\mid \hat{K_i}\phi
				\mid [\pi] \phi
				\mid \ldiamondarg{\pi}\phi
				% 	   \mid [!\pi]\phi
				% \pi  &
				%\ep\mid
				% \dl\mid	
				%a
				% 	  \mid ?\phi_b
				%\mid \pi\cdot \pi
				% \mid \pi + \pi
				% \mid \pi^*\\
\end{array}$$
where $p\in\BP$, $i\in \Ag$ and $\pi\in\LL_{obs}$.
\end{definition}
In simpler words, a $\POL$ formula is in a Negative Normal Form if and only if the negations ($\neg$) occur in the formula only before a propositional letter. Now we can define formally the various fragments of $\POL$.

\begin{definition}[$\starfree$ Formula ans $\starfree$ fragment of $\POL$]
Given a set of propositional letters $\BP$, the $\starfree$ formulas of $\POL$ is defined recursively as following:
$$\begin{array}{r@{\quad::= \quad}l}
				\phi  &
				\top
				\mid
				p
				\mid \neg \phi
				\mid \phi \land \phi
				\mid K_i\phi
				\mid [\pi] \phi
				% 	   \mid [!\pi]\phi
				% \pi  &
				%\ep\mid
				% \dl\mid	
				%a
				% 	  \mid ?\phi_b
				%\mid \pi\cdot \pi
				% \mid \pi + \pi
				% \mid \pi^*\\
			\end{array}$$
where $p\in\BP$, $i\in \Ag$ and $\pi\in\LL_{obs}$ is defined recursively as:
$$\begin{array}{r@{\quad::= \quad}l}
				\pi  &
				a\in\Sigma
				\mid \ep
				\mid \pi + \pi
				\mid \pi.\pi
\end{array}$$

The $\starfree$ Fragment of $\POL$ consists only $\starfree$ formulas of $\POL$.
\end{definition}

% \begin{definition}[$\starfree$ Fragment]
% The $\starfree$ Fragment of $\POL$ is the part of $\POL$ where all the regular expressions $\pi$ in the regular modalities $\ldiamondarg{\pi}\phi$ and $[\pi]\phi$, where $\phi$ is a $\starfree$ Fragment formula, is recursively defined as:
% $$\begin{array}{r@{\quad::= \quad}l}
% 				\pi  &
% 				a\in\Sigma
% 				\mid \ep
% 				\mid \pi + \pi
% 				\mid \pi.\pi
% \end{array}$$
% \end{definition}

Hence the $\starfree$ Fragment of $\POL$ just consists of formulas where the regular expression in the modalities does not contain any Kleene-Star in it.

\begin{definition}[$\existential$ Formula and $\existential$ fragment of $\POL$]
Given a set of propositional letters $\BP$, the $\existential$ formulas of $\POL$ can be inductively defined as:
$$\begin{array}{r@{\quad::= \quad}l}
				\phi  &
				\top
				\mid
				p
				\mid \neg p
				\mid \phi \land \phi
				\mid \hat{K_i}\phi
				\mid \ldiamondarg{\pi}\phi
				% 	   \mid [!\pi]\phi
				% \pi  &
				%\ep\mid
				% \dl\mid	
				%a
				% 	  \mid ?\phi_b
				%\mid \pi\cdot \pi
				% \mid \pi + \pi
				% \mid \pi^*\\
\end{array}$$
where $p\in\BP$, $i\in \Ag$ and $\pi\in\LL_{obs}$.

The $\existential$ Fragment of $\POL$ consists only $\existential$ formulas of $\POL$.
\end{definition}

That is, the formulas of $\existential$ Fragments are all in NNF and only contains $\hat{K_i}$ and $\ldiamondarg{\pi}$ operators for an agent $i$ and a regular expression $\pi$.

% \begin{definition}[$\existential$ Fragment]
% The formulas in the $\existential$ Fragment of $\POL$ is such that both the following conditions should satisfy: 
% \begin{itemize}
%     \item The formulas are in NNF.
%     \item The formula only has operators $\hat{K_i}$ and $\ldiamondarg{\pi}$ in it, for an agent $i\in\Ag$ and a regular expression $\pi$.
% \end{itemize}
% \end{definition}

\begin{definition}[$\starfree-\existential$ Formula and $\starfree-\existential$ fragment of $\POL$]
Given a set of propositional letters $\BP$, the $\existential$ formulas of $\POL$ can be inductively defined as:
$$\begin{array}{r@{\quad::= \quad}l}
				\phi  &
				\top
				\mid
				p
				\mid \neg p
				\mid \phi \land \phi
				\mid \hat{K_i}\phi
				\mid \ldiamondarg{\pi}\phi
				% 	   \mid [!\pi]\phi
				% \pi  &
				%\ep\mid
				% \dl\mid	
				%a
				% 	  \mid ?\phi_b
				%\mid \pi\cdot \pi
				% \mid \pi + \pi
				% \mid \pi^*\\
\end{array}$$
where $p\in\BP$, $i\in \Ag$ and $\pi\in\LL_{obs}$, where $\pi$ is inductively defined as:
$$\begin{array}{r@{\quad::= \quad}l}
				\pi  &
				a\in\Sigma
				\mid \ep
				\mid \pi + \pi
				\mid \pi.\pi
\end{array}$$

The $\starfree-\existential$ Fragment of $\POL$ consists only $\starfree-\existential$ formulas of $\POL$.
\end{definition}

\begin{definition}[$\word$ Formula and $\word$ fragment of $\POL$]
Given a set of propositional letters $\BP$, the $\word$ formulas of $\POL$ is defined recursively as following:
$$\begin{array}{r@{\quad::= \quad}l}
				\phi  &
				\top
				\mid
				p
				\mid \neg \phi
				\mid \phi \land \phi
				\mid K_i\phi
				\mid [\pi] \phi
				% 	   \mid [!\pi]\phi
				% \pi  &
				%\ep\mid
				% \dl\mid	
				%a
				% 	  \mid ?\phi_b
				%\mid \pi\cdot \pi
				% \mid \pi + \pi
				% \mid \pi^*\\
			\end{array}$$
where $p\in\BP$, $i\in \Ag$ and $\pi\in\LL_{obs}$ is defined recursively as:
$$\begin{array}{r@{\quad::= \quad}l}
				\pi  &
				a\in\Sigma
				\mid \ep
				\mid \pi.\pi
\end{array}$$

The $\word$ Fragment of $\POL$ consists only $\word$ formulas of $\POL$.
\end{definition}

% \begin{definition}[$\word$ Fragment]
% The $\word$ Fragment of $\POL$ is the part of $\POL$ where all the regular expressions $\pi$ in the regular modalities $\ldiamondarg{\pi}\phi$ and $[\pi]\phi$, where $\phi$ is a $\word$ Fragment formula, is recursively defined as:
% $$\begin{array}{r@{\quad::= \quad}l}
% 				\pi  &
% 				a\in\Sigma
% 				\mid \ep
% 				\mid \pi.\pi
% \end{array}$$
% \end{definition}
Therefore, in other words, the $\word$ Fragment of $\POL$ contains only $\POL$ formulas where the regular expressions in the modalities are just words over $\Sigma$.
		
		\begin{definition}[Truth definition] Given an epistemic expectation model $\M$ = $\langle S, \sim ,V, \Exp\rangle$, a world $s\in S$, and a $\POL$-formula $\phi$, the truth of $\phi$ at $s$, denoted by $\M,s\vDash \phi$, is defined by induction on $\phi$ as follows: 
			$$
			\begin{array}{rcl}
				%& & \\
				\M,s\vDash p &\Leftrightarrow& p\in V(s)\\
				\M,s\vDash \neg\phi &\Leftrightarrow&   \M,s\nvDash \phi \\
				\M,s\vDash \phi\land \psi &\Leftrightarrow& \M,s\vDash \phi \textrm{ and } \M,  s\vDash \psi \\
				\M,s\vDash K_i\phi &\Leftrightarrow& \textrm{for all }t: (s\sim_i t  \textrm{ implies } \M,t\vDash\phi)\\
				%\M,s\vDash [\pi]\phi &\Leftrightarrow& \textrm{for all }w\in\LL(\pi): (w\in\prefixes(\Exp(s))  \\
				%&& \textrm{ implies }\M|_w,s\vDash \phi)\\ & &\\
				\M,s\vDash [\pi]\phi &\Leftrightarrow& \textrm{for all }w\in\LL(\pi) \cap \prefixes(\Exp(s))  \\
				&& \textrm{ we have }\M|_w,s\vDash \phi)
			\end{array}
			$$
			%\todo{François changed the presentation of semantics of $[\pi]$ so that the box has no overfull}

		\end{definition}
		
		The truth of $K_i\phi$ at $s$ follows the standard possible world semantics of epistemic logic. The formula $[\pi]\phi$ holds at $s$ if for every observation $w$ in the set $\LL(\pi)$  that matches with the beginning  of (i.e., is a prefix of) some expected observation in $s$, $\phi$ holds at $s$ in the updated model $\M|_w$. Note that $s$ is a world in $\M|_w$ because $w \in \prefixes(\Exp(s))$. Similarly, the truth definition of $\ldiaarg{\pi}\phi$ can be given as follows: $\M,s\vDash \ldiaarg{\pi}\phi \textrm{ iff there exists }w\in\LL(\pi) \cap \prefixes(\Exp(s)) 
		\textrm{ such that }\M|_w,s\vDash \phi$. Intuitively, the formula $\ldiaarg{\pi}\phi$ holds at $s$ if there is an observation $w$ in $\LL(\pi)$  that matches with the beginning  of some expected observation in $s$, and $\phi$ holds at $s$ in the updated model $\M|_w$. For the examples described earlier, we have:
		
		\begin{itemize}
			\item[-] $\M_{tl}, s \models [g^*]\neg(K_T f \lor K_T \neg f)$. This example corresponds to a safety property: there is no leak of information when observing an arbitrary number of $g$'s because it is compatible with both the expectation $g^*ar^*)^*$ of the French traffic light system, and the expectation $g^*ar^*a)^*$ of the non-French one.
			
			\item[-] $\M_{tl}, s \models \ldiamondarg{(garg)^*}(K_T f)$. This example in the $\existential$ fragment shows that we can express the existence of a sequence of observations that reveals that the traveller is in France.
			
			\item[-] $\M_{tl}, s \models \ldiamondarg{gar}\neg(K_T f \lor K_T \neg f)$. This example in the $\word$ fragment expresses that the sequence of observations $gar$ would keep the traveller ignorant about her whereabouts.  
			
			\item[-] $\M_{mp}, s \models \ldiamondarg{m}((K_R d \land \neg K_A d)$. This example, also in the $\word$ fragment, expresses that after receiving the message $m$, the receiver gets to know about the decision but the attacker remains ignorant.
		\end{itemize}
		
		%The last two examples are in the $\word$ fragment and show that we can check whether some sequence of observations leads to some epistemic property.

		\noindent\textbf{Model Checking for $\POL$: } %We now formally describe the model checking problem for $\POL$ and its fragments. 
		Given a finite pointed epistemic expectation model $\M,s$, and a formula $\phi$, does $\M,s \models \phi$? We are interested in knowing the complexity of this problem. %The model checking for various fragments of $\POL$ (like $\word$, $\starfree$, $\existential$ and $\starfree\existential$) is defined likewise, where the only extra promise is in the structure of the formula $\phi$. 
		We will also consider restrictions of the model checking when the input formula $\phi$ is restricted to be in one of the syntactic fragments: $\word$, $\starfree$, $\existential$ and $\starfree\existential$.
		
		%Again consider the Amazon scenario explained before. For the model $M = \ldiamondarg{S,\{R_{sender}\}, V, Exp}$, as explained earlier, consider the formula $\ldiamondarg{a}\hat{K_{sender}}p$, that is, after observing success message, the sender knows the package is delivered. Note that $\mathcal{L}(b\backslash a) = \emptyset$. Hence the state $s$ survives in $\M|_a$, which results in the survival of only the reflexive relation $R_{sender} = \{(s,s)\}$. Hence $\M|_a, s\vDash\hat{K_{sender}}p$. Hence $\M,s\vDash\ldiamondarg{a}\hat{K_{sender}}p$
		
		%\section{POL model checking problem}
		
		%\todo{we have a problem of vocabulary: "state" in models and "state" in automata. Is "worlds" for models ok?}

\section{Complexity results}\label{results}

% The main complexity result that we prove in this paper is the following:
The main complexity result that we prove is given below. 
For all the proof details, see the appendix.

\begin{restatable}{theorem}{POLPSPACE}\label{theorem:full}
    $\POL$ model checking is $\PSPACE$-complete.
\end{restatable}

%\

\noindent\textbf{$\POL$ model checking is in $\PSPACE$:} For proving the upper bound result, that is, showing that $\POL$ model checking is in $\PSPACE$, we design the algorithm $\DecidePSPACE$ (Algorithm~\ref{algoMAIN:PSPACE}). It takes as input a $\POL$ model $\M = \ldiamondarg{S,\sim,V,\Exp}$, an initial starting world  $s\in S$, and a $\POL$ formula $\varphi$ and returns $\T$ iff $\M, s \models \phi$. 
We also prove that the algorithm $\DecidePSPACE$ uses polynomial space. 

%$\DecidePSPACE$ is a recursive algorithm and is divided
The recursive algorithm $\DecidePSPACE$ is divided
into various cases depending on the structure of $\phi$. The subtle case is the observation modality $\ldiaarg \pi \psi$ (that is dealt with in lines~\ref{ln:ifspi} to \ref{ln:ifepi}).  It follows from the truth definition that $\M, s\vDash \ldiaarg \pi \psi$ iff there exists a $w\in \LL(\pi)$ such that $\M|_w, s\vDash \psi$. 
Here we observe that for any $\M$ and $w$ the model $\M|_w$ can be represented by a string of size polynomial in the size of $\M$ (This is because $\M$ and $\M|_w$ just differ by their expected observation functions as follows: for any world $t$,  $\Exp'(t) = \Exp(t)\backslash w$ and $\Exp(t)$ share the same NFA, just the set of initial states is different.).
Thus if we consider the set $\Gamma^\M = \{\M|_w\mid w\in\Sigma^*\}$, that is, the set of every updated model $\M|_w$, for a $\POL$ model $\M$, over all $w\in\Sigma^*$, we realize that all the models in $\Gamma^{\M}$ has size polynomial in the size of $\M$. 
Thus, by using both the observations together, when $\DecidePSPACE$ has to check if $\M, s\vDash \ldiaarg \pi \psi$ (in the \textbf{for} loop in lines~\ref{ln:forsoracle} to \ref{ln:foreoracle}) it goes over all models $\M'$ in $\Gamma^{\M}$ and (in line~\ref{line:forallmodelstart}) checks if $\M'=\M|_w$ for some $w\in \LL(\pi)$ and finally (in line~\ref{ln:recursive}) calls $\DecidePSPACE$ recursively to check if $\M|_w, s\vDash \psi$.

%The notation $t\in\M$ for a world $t$ is used informally to denote $t\in S$, where $S$ is set of worlds in $\M$. 

% Consider $\Gamma^\M = \{\M|_w\mid w\in\Sigma^* \}$, that is, it is the set of every updated model $\M|_w$, for a $\POL$ model $\M$, over all $w\in\Sigma^*$. We observe that any element $\M'$ in $\Gamma^\M$ can be represented by a string of size polynomial in $\M$. Indeed, $\M$ and $\M'$ just differ by their expected observation functions as follows: for any world $t$,  $\Exp'(t)$ and $\Exp(t)$ share the same NFA, just the current subset of states is different\footnote{Actually, the for loop at line 11 in the algorithm is a loop over all strings of size $poly(|\M|)$ so that all elements $\M'$ in $\Gamma^\M$ will be represented.}. The search for a candidate updated model $\M'$ reached some word in $\LL(\pi)$ is performed by a $\PSPACE$ oracle.

Thus $\DecidePSPACE$ needs to call a polynomial space subroutine to check if $\M' = \M|_w$ for some word $w \in \LL(\pi)$. 
To prove that there exists such a polynomial space algorithm 
%that is in $\PSPACE$ and checks if $\M' = \M|_w$ for some word $w \in \LL(\pi)$, 
we present a slightly convoluted argument. 
Algorithm~\ref{algoMAIN:algofororacle} provides a non-deterministic procedure running in polynomial space for deciding that $\M' = \M|_w$ for some word $w \in \LL(\pi)$.  By Savitch's theorem \cite{DBLP:journals/jcss/Savitch70} which states that $\mathsf{NPSPACE}$ = $\PSPACE$, we have that a polynomial space algorithm also exists. Algorithm~\ref{algoMAIN:algofororacle} starts by guessing a word of exponential length, sufficiently long enough to explore all subsets of current states for NFAs of  $\Exp(t)$ for all worlds $t$ in $\M$ and for the NFA of $\pi$. Then the algorithm guesses the word $w$ letter by letter and it progresses in the NFAs (note that it does not store the word $w$ as it can be of exponential length).  Algorithm~\ref{algoMAIN:algofororacle} accepts when $w \in \LL(\pi)$ (i.e., $\emptyword \in \LL(\pi')$) and $\M = \M'$. Otherwise, it rejects. 	

%In Theorem~\ref{thm:correctness} we prove that $\DecidePSPACE$ properly checks whether $M, s\vDash \varphi$. Finally in Theorem~\ref{thm:space} we prove that $\DecidePSPACE$ runs in polynomial space. 

% \todo{replace some "$<$" by "$\langle$ etc.}

\begin{algorithm}[tb]
	\caption{$\DecidePSPACE$\label{algoMAIN:PSPACE}}
	\textbf{Input}: $\M = \ldiamondarg{S,\sim,V,\Exp}, s\in S, \varphi$\\
	\textbf{Output}: $\T$ iff $\M,s\vDash \varphi$
	\begin{algorithmic}[1] %[1] enables line numbers
		\IF{$\varphi=p$ is a propositional variable }
		\STATE return $\T$ if $p \in V(s)$; $\Fa$ otherwise
		\ENDIF
		\IF{$\varphi = \neg\psi$}
		\STATE return not $\DecidePSPACE(\M,s,\psi)$
		\ENDIF
		\IF{$\varphi = \psi'\vee \psi$}
		\STATE return $\DecidePSPACE(\M, s, \psi)$ or $\DecidePSPACE(\M, s, \psi')$
		\ENDIF
		\IF{$\varphi = \ldiaarg\pi\psi$} \label{ln:ifspi}
		\FOR {all models $\M'$ in $\Gamma^{\M}$}\label{line:forallmodelstart} \label{ln:forsoracle}
%		\IF{$\reach(\M, \M', \pi)$ is $\T$} 
		\IF{$s$ is a world in $\M'$ and the oracle claims that $\M' = \M|_w$ for some word $w \in \LL(\pi)$}\label{line:oraclecall} \label{ln:oracle}
		\STATE return $\DecidePSPACE(\M', s, \psi)$ \label{ln:recursive}
		\ENDIF
		\ENDFOR\label{line:forallmodelend} \label{ln:foreoracle}
		\STATE return $\Fa$
		\ENDIF \label{ln:ifepi}
		\IF{$\varphi = \hat{K_i}\psi$}  	
		\IF{$\exists t\in S$ such that $t\sim_i s$ and $\DecidePSPACE(\M, t, \psi)$} \label{line:K_i}
%		\IF{there exists a world $t\in S$ such that $t\sim_i s$ and $\DecidePSPACE(\M, t, \psi)$} \label{line:K_i}
		\STATE  return $\T$
		\ELSE
		\STATE  return $\Fa$
		\ENDIF
		\ENDIF
	\end{algorithmic}
\end{algorithm}

\begin{algorithm}[tb]
	\caption{Non-deterministic  procedure  to  decides  that $\M' = \M|_w$ for some word $w \in \LL(\pi)$\label{algoMAIN:algofororacle}}
	\textbf{Input}: $\M = \ldiamondarg{S,\sim,V,\Exp}, \M'\in\Gamma^\M, \pi$\\
	\textbf{Output}: has an accepting execution  iff $\M' = \M|_w$ for some $w \in \LL(\pi)$
	\begin{algorithmic}[1] %[1] enables line numbers	    
  \STATE$\pi' := \pi$
	\FOR{ $i = 1$ to $2^{\pi}\times \Pi_{t \in S} 2^{|\Exp(t)|}$ }\label{ln:oracleforstart}
	\IF{$\emptyword \in \LL(\pi')$ and $\M = \M'$}
	\STATE \algoaccept
	\ENDIF
	\STATE guess a letter $a$ from $\Sigma$ 
	\STATE $\pi' :=\pi'\regdiv a$
	\FOR{ each world $t$ in $S$ }
	\STATE $\Exp(t) := \Exp(t) \regdiv a$\hfill // we modify $\M$ locally
	\ENDFOR
	\ENDFOR
    \STATE	\algoreject
	\end{algorithmic}
\end{algorithm}

% %\section{Various Fragments of $\POL$}\label{frag}
%\

\noindent\textbf{Model checking  for $\POL$ is $\PSPACE$-hard:}
Interestingly, there are two sources for the model checking to be $\PSPACE$-hard: Kleene star in observation modalities as well as alternations in modalities (sequences of nested existential and universal modalities). 
%In Theorem~\ref{theorem:existential} and Theorem~\ref{theorem:starfree} 
We prove the $\PSPACE$-hardness of model checking against the $\existential$ fragment  and the $\starfree$ fragment of $\POL$ respectively.

\begin{restatable}{theorem}{existentialLB}
The model checking for the $\existential$ fragment of $\POL$ is $\PSPACE$-hard.\label{theorem:existential}
\end{restatable}
\begin{restatable}{theorem}{starfreeLB}
	The model checking for $\POL$ is $\PSPACE$-hard,  when the $\POL$ formulas are $\starfree$.\label{theorem:starfree}
\end{restatable}

%\ 

\noindent\textbf{Model checking for $\starfree$ $\existential$ and $\word$ fragment of $\POL$:} While Theorems~\ref{theorem:existential} and \ref{theorem:starfree} proved the $\PSPACE$-hardness of the model checking for the $\existential$ fragment and the $\starfree$  fragment of $\POL$, respectively, 
if we consider the $\starfree$ $\existential$ fragment then we can show that the model checking is $\NP$-complete.  Finally, we also prove that the model checking for the $\word$ fragment is in $P$.

\begin{restatable}{theorem}{starfreeexistNPC}
	The model checking problem for the $\starfree$ $\existential$ fragment of $\POL$ is $\NP$-complete.\label{theorem:starfreeexist}
\end{restatable}

\begin{restatable}{theorem}{wordP}\label{theorem:word}
Model checking for the $\word$ fragment is in $\PTime$.
%	The model checking problem for $\POL$, when $\pi$'s are words, is in $\PTIME$. 
\end{restatable}

\section{Application}\label{app}
% % !TeX spellcheck = en_US
% \newcommand{\obsright}{\blacktriangleright}
% \newcommand{\obsleft}{\blacktriangleleft}
% \newcommand{\obsup}{\blacktriangle}
% \newcommand{\obsdown}{\blacktriangledown}

% \newcommand{\expwater}{(\obsright \union \obsup)^* (\obsdown \union \obsleft \union \ep) (\obsright \union \obsup)^*}
% \newcommand{\exppower}{(\obsleft \union \obsdown)^* (\obsup \union \obsright \union \ep) (\obsleft \union \obsdown)^*}
% \newcommand{\exppatrol}{(\obsright^+ \obsdown^+ \obsleft^+ \obsup^+)^*}
Let us consider an automatic farming drone that is moving in a field represented as a grid (see Figure~\ref{figure:field}). Two agents $a$ and~$b$ help the farming drone. The system is adaptive so the global behaviour is not hard-coded but learned. We suppose that the drone moves on a grid and agents $a$ and $b$ may observe one of the four directions: $\Act := \set{\obsright, \obsleft, \obsup, \obsdown}.$ For instance, observing $\obsleft$ means that the drone moves one-step left. For this example, we suppose that agent $a$ has learned that there are three possible expectations for the drone:
\begin{enumerate}
	\item the drone may go up-right searching for water, but the drone can make up to one wrong direction ($\obsdown$ or $\obsleft$). The corresponding set of expectations is captured by the regular expression
	$\expwater$ where $\ep$ stands for the empty word regular expression.
	\item the drone may go down-left searching for power supply, but the drone can make up to one wrong direction ($\obsup$ or $\obsright$). The corresponding set of expectations is captured by the regular expression:
	$\exppower$.
	\item the drone is patrolling making clockwise squares. The expectation is:
	$\exppatrol$.
\end{enumerate}

The regular expressions may be learned by the agents after observing several executions (see for instance \cite{DBLP:conf/birthday/BalcazarDG97}) or might be computed by planning techniques \cite{DBLP:conf/aips/BonetPG09}.
Agent $b$ has more information and knows that the behaviour of the drone would include either searching for water or power supply. Agent $a$ is programmed so that if she knows that the drone is searching for water ($K_a water$) then she will turn on the valve, and if she knows that the drone is searching for power  ($K_a power$), she will prepare the power supply. Agent $b$ is programmed in the same way.
The model $\modelM$, depicted in Figure~\ref{figure:motivationalexamplekripkemodel}, can be obtained by techniques described in~\cite{van2014hidden} (they use a mechanism from $\DEL$ for constructing the epistemic expectation model, by assigning the expectations at each world).

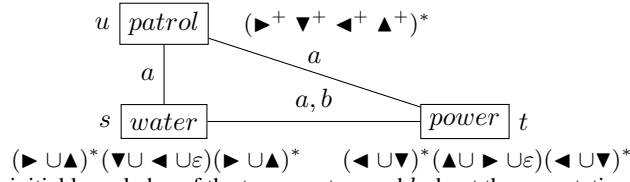
\begin{figure}
	\begin{center}
		\begin{tikzpicture}[yscale=1.3]
			\node[world] (s) {$water$};
			\node at (-0.1, -0.4) {\small $\expwater$};
			\node[world] (t) at (4, 0) {$power$};
			\node at (4.3, -0.4) {\small $\exppower$};
			\node[world] (u) at (0, 1) {$patrol$};
			\node at (2.3, 1) {\small $\exppatrol$};
			\node[left = 0mm of s] {$s$};
			\node[right = 0mm of t] {$t$};
			\node[left = 0mm of u] {$u$};
			\draw (s) edge node[above] {$a,b$} (t);
			\draw (s) edge node[left] {$a$} (u);
			\draw (t) edge node[above] {$a$} (u);
		\end{tikzpicture}
	\end{center}
	\vspace{-5mm}
\caption{Model describing the initial knowledge of the two agents~$a$ and $b$ about the expectation of the automatic farming drone.\label{figure:motivationalexamplekripkemodel}}
\end{figure}

 %If agent $a$ is programmed so that if she knows that the drone is searching for power  ($K_a power$), she will prepare the power supply.  Agent $b$ is programmed in the same way.

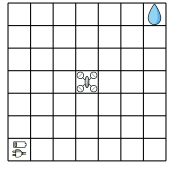
\begin{figure}
	\begin{center}
		\newcommand{\sizefield}{7}
		\begin{tikzpicture}[scale=0.3]
			\foreach \x in {0, 1, ..., \sizefield} {
				\draw (\x, 0) -- (\x, \sizefield);
				\draw (0, \x) -- (\sizefield, \x);
			}
			\node at (0.5, 0.5) {\includegraphics[width=0.3cm]{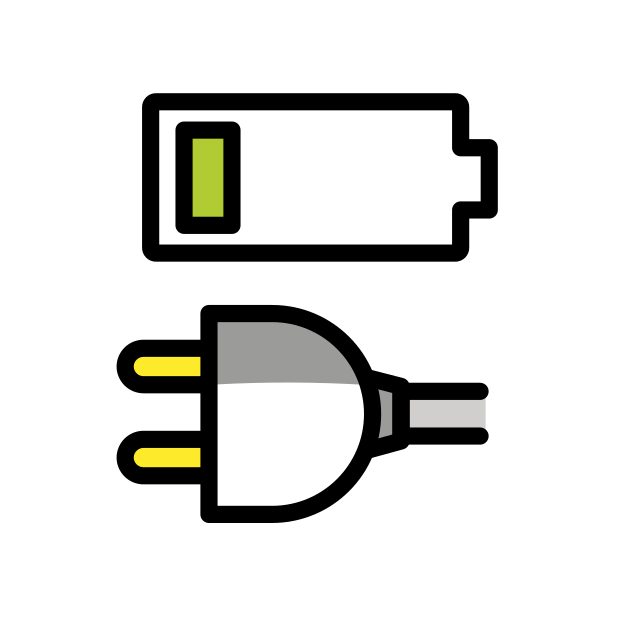}};%power
			\node at (6.5, 6.5) {\includegraphics[width=0.3cm]{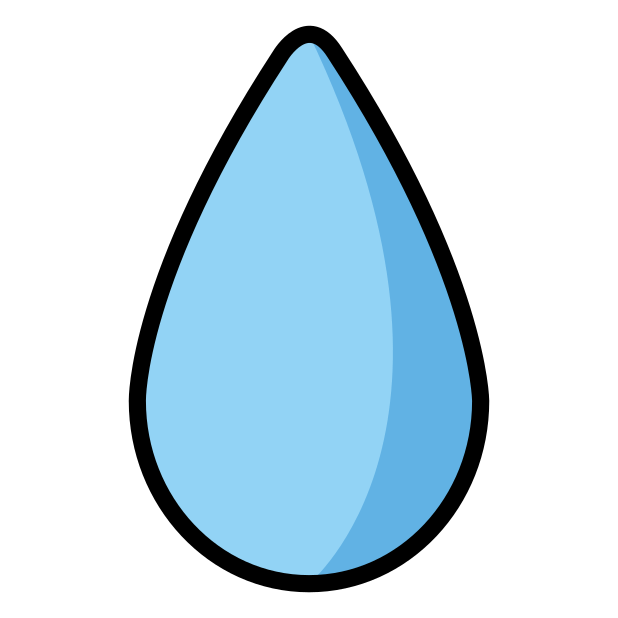}};%water
			\node at (3.5, 3.5) {\includegraphics[width=0.3cm]{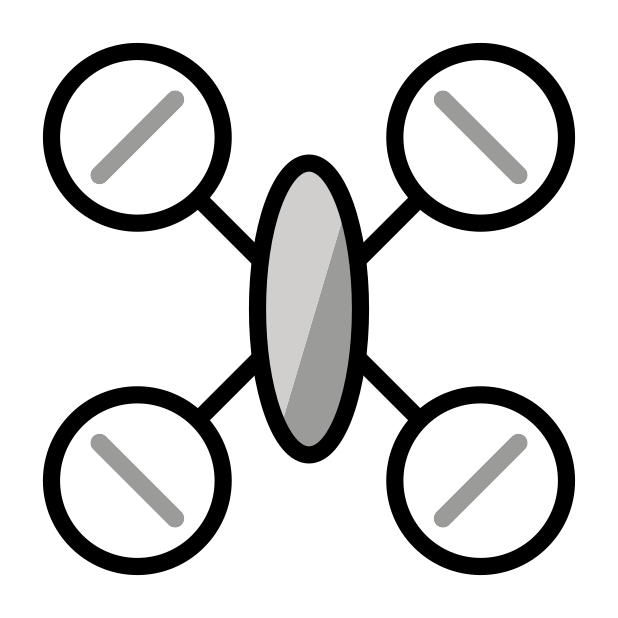}};%drone
		\end{tikzpicture}
	\end{center}
	\vspace{-3mm}
\caption{Field and an automatic farming drone.\label{figure:field}}
\end{figure}

\newcommand{\fragmentexample}[1]{\ensuremath{#1} fragment}

Now, verification tasks related to epistemic planning (e.g., verifying whether $K_a water$, $K_b water$,
 $K_a power$, or $K_b power$ is true after some observations) reduce to the $\POL$ model checking problem. Let us now discuss the expressivity of the fragments: $\word$, $\existential$, $\starfree-\existential$ and $\starfree$.
 %\begin{center}
 %\begin{tabular}{l|l}
 %Fragment & Verification task \\
 %\hline
 %$\word$ & Verification of a plan \\
 %$\existential$ & Epistemic planning \\
 %$\starfree \existential$ & Bounded epistemic planning \\
 %$\starfree$ & Generalized bounded epistemic planning \\
 %Full language  & Generalized epistemic planning
 %\end{tabular}
 %\end{center}
 
 In the $\word$ fragment, words are fixed sequences of observations. The fragment thus enables to write formulas of the form $\ldiaarg{w}\phi$, meaning that $\phi$ holds after the sequence $w$ of observations (that can be considered as the observations produced by the plan executed by the system). Thus, this fragment enables to write formulas to verify properties after the execution of a plan.
 
 \begin{example}[verification of a plan, \fragmentexample{\word}]
 	Does agent $a$ know that the drone is searching for water after the sequence $\obsright  \obsright \obsright$? 
	$$\M, s \models \ldiaarg{\obsright  \obsright \obsright} K_a water$$
 \end{example}
 
 Epistemic planning is the general problem of verifying whether there exists a plan leading to a state satisfying a given epistemic formula. In our setting, it can be expressed by a formula of the form $\ldiaarg{\pi}\phi$ where $\pi$ denotes the plan search space (more precisely the search space of sequences of observations produced by a plan).
 
 \begin{example}[epistemic planning, \fragmentexample{\existential}]
 Does there exist a plan for the drone such that agent $b$ would know that the drone is searching for water while agent $a$ would still consider patrolling a possibility?
	$$\M, s \models \ldiaarg{(\obsright \union \obsdown \union \obsleft \union \obsup)^*} (K_b water \land \hat K_a patrolling)$$	
 \end{example}
 
In planning (and also in epistemic planning), we may ask for the existence of a plan of bounded length, e.g., less than 4 actions. The \fragmentexample{\starfree \existential} is sufficiently expressive to tackle the so-called \emph{bounded} epistemic planning. 
			
\begin{example}[bounded epistemic planning, $\starfree$ \fragmentexample{\existential}] Does there exist a sequence of at most 4 moves such that agent $b$ would know that the drone is searching for water while agent $a$ would still consider patrolling a possibility?
	$$\M, s \models \ldiaarg{(\obsright \union \obsdown \union \obsleft \union \obsup \union \epsilon)^4} (K_b water \land \hat K_a patrolling)$$
\end{example}

Interestingly, the \fragmentexample{\starfree} and the full language are able to express properties, mixing existence and non-existence of plans, in respectively the bounded and unbounded cases.

\begin{example}[\fragmentexample{\starfree}]
Agent $a$ would not gain the knowledge that the drone will search for water with less than or equal to 2 movements but it is possible with 3 movements:
		\begin{align*}
\M, s \models & [(\obsright \union \obsdown \union \obsleft \union \obsup)^2] \lnot K_a water ] 
 \land  \\ &\ldiaarg{(\obsright \union \obsdown \union \obsleft \union \obsup)^3} K_a water
		\end{align*}
\end{example}

\begin{example}[full language]
	It is impossible for the agent $a$ to know that the drone is searching for water with only down and left movements but there is a plan if all movements are allowed:
		\begin{align*}
		\M, s \models & [(\obsdown \union \obsleft)^*] \lnot K_a water ] 
		\land  \\ &\ldiaarg{(\obsright \union \obsdown \union \obsleft \union \obsup)^*} K_a water
	\end{align*}
\end{example}

% \paragraph{On implementation.}
% The model checking for the $\starfree \existential$ fragment can be implemented via a reduction to SAT.  (see Details in the Supplementary Material).
%the PSPACE-complete fragments,  a reduction to QBF has to be investigated.

\subsection{On Implementation}
The model checking for the $\word$ fragment can be implemented in poly-time with a bottom-up traversal of
the parse tree of the formula, as for CTL \cite[Section 6.4]{DBLP:books/daglib/0020348}. In this subsection, we explain how to provide an efficient implementation for the $\starfree-\existential$ fragment of POL model checking by providing a reduction to SAT. We explain how to check that $\M, w \models \ldiaarg{\pi}\phi$, where $\pi$ is star-free and $\phi$ is an epistemic formula.  In other words, we aim at checking whether there is a guessed word that belongs to the language of $\pi$ such that $w$ survives the announcement of that word and $\phi$ holds in the updated model. As $\pi$ is star-free, the guessed word is bounded by the size of $\pi$. The idea is to introduce propositional variables to model the Boolean values of the following statements: (i) the $t$-th letter of the guessed word is equal to $a$, (ii) a given automaton $A$ is in state $q$ after having read the  first $t$ letters of the guessed word, and, (iii) a subformula of the formula $\phi$ to check is true at a given world $u$. The last type of statements are combined in the spirit of the Tseitin transformation \cite[p. 91]{DBLP:books/daglib/0029942}.  We leave the other cases for future work.
\par W.l.o.g. we suppose that the guessed word is of length~$k$. We also suppose that all the automata are deterministic. This costs an exponential time in principle, but we can rely on an efficient minimization algorithm to obtain small deterministic automata in practical cases.

\newcommand{\existenceunicity}[1]{\#(#1)=1}
\newcommand\propositionletter[2]{p_{#1, #2}}
\newcommand\propositiontruth[2]{t_{#1, #2}}
\newcommand\propositionautomatonstate[3]{a_{#1, #2, #3}}
\newcommand{\lequiv}{\leftrightarrow}

We now explain how to define a Boolean formula $tr(M, w, \ldiaarg{\pi}\phi)$ such that $\M, w \models \ldiaarg{\pi}\phi$ iff $tr(M, w, \ldiaarg{\pi}\phi)$ is a satisfiable Boolean formula.
To this end, we introduce several propositional variables:
\begin{itemize}
    \item $\propositionletter t a$: the $t$-th letter of the guessed word is $a$
    \item $\propositiontruth u \psi$: subformula $\psi$ (of $\phi$) holds in world $u$ in the updated model
    \item $\propositionautomatonstate A t q$: in the automaton $A$, the state after having read $t$ letters in the guessed word is $q$. Note that $A$ can denote either an automaton for $\pi$ or that for any prefix language $\Exp(u)$.
\end{itemize}

Given a finite set $P$ of propositional variables, we write $\existenceunicity P$ for a Boolean formula saying that exactly one propositional variable in $P$ is true.
We construct the following Boolean formula 
$$\propositiontruth w \phi \land guessedWord \land good \land surv \land rules $$
which is satisfiable iff $\M, w \models \ldiaarg{\pi}\phi$. 
The first part $\propositiontruth w \phi$ enforces that the formula $\phi$ is true in $w$ after the announcement of the guessed word. 

The second part $guessedWord$ says that the guessed word is uniquely determined by the propositions $\propositionletter t a$. More precisely, that part is:
$\bigwedge_{t=1..k}\existenceunicity{\set{\propositionletter t a \mid a \in \Sigma}}$.

In other words, it means that there is a path in the non-deterministic automaton $A$ corresponding to $\pi$, starting from the initial state $q_0$,  following the guessed word, and leading to a final state.    Given $A$, the automaton for $\pi$, the formula $good$ is the conjunction of:
\begin{itemize}
    \item $\propositionautomatonstate A 0 {q_0}$;
    \item $\bigwedge_{t=0..k} \existenceunicity{\set{\propositionautomatonstate A t q \suchthat q \in Q}}$;
    \item $\bigwedge_{t=0..k-1, q, a} \propositionautomatonstate A t q \land \propositionletter t a  \rightarrow
    \propositionletter A {t+1} {\delta(q, a)}$
    %\bigvee_{q'  \in \delta(q, a)} \propositionletter A {t+1} {q'}$;
    \item $\bigvee_{q \in F} \propositionautomatonstate A k q$.
\end{itemize}

The part $surv$ says that the guessed word belongs to the prefix language of $\Exp(s)$. That formula is similar to $good$ but for the automaton corresponding to the prefix language of $\Exp(s)$.
Finally, $rules$ is a formula that mimics the semantics of $\propositiontruth w \phi$ in the same spirit as Tseitin transformation.
%already cited in the text
Formula $rules$ is the conjunction of:
\begin{itemize}
    \item $\propositiontruth u p \lequiv \top$ if $u \models p$;
    \item $\propositiontruth u p \lequiv \bot$ if $u \models p$;
    \item $\propositiontruth u {\lnot \psi} \lequiv \lnot \propositiontruth u {\psi}$;
    \item  $\propositiontruth u {\psi_1 \land \psi_2} \lequiv  (\propositiontruth u  {\psi_1} \land \propositiontruth u  {\psi_2})$;
    \item  $\propositiontruth u {K_i \psi} \lequiv \bigwedge_{v \mid u \rightarrow_i v} surv(v) \rightarrow \propositiontruth u  \psi$
\end{itemize}

\noindent where $surv(v)$ is a propositional variable saying that world $v$ survives the announcement of the guessed word. The variable $surv(v)$ alone is not sufficient. It is accompanied by a collection of clauses in the same spirit as $surv$ but with the automaton of the prefixes of $\Exp(v)$.

We have $\M, w \models \ldiaarg{\pi}\phi$ iff $tr(M, w, \ldiaarg{\pi}\phi)$ is a satisfiable Boolean formula. Moreover, the truth values of propositions $\propositionletter t a$ in a valuation satisfying $tr(M, w, \ldiaarg{\pi}\phi)$ gives a plan, i.e. the guessed word in the language of $\pi$ such that $\phi$ holds after executing that word from $\M, w$.
The full case of $\starfree-\existential$ fragment follows the same idea but is cumbersome.

The reduction to SAT has been implemented in \texttt{Python3} using the library \texttt{pySAT} (with the SAT solver \texttt{Glucose}) and the library \texttt{automata-lib}. To get an idea, the running time for checking Example 3 is around 10ms.\footnote{The accompanying codes can be found in the following link: https://github.com/francoisschwarzentruber/polmc}

\section{Related work}\label{rel}
%\noindent\textbf{Related Works: } 

%It is interesting to note that 
%\sujata{We note that} the model checking problem for $\EL$ is in $\PTime$, whereas for $\DEL$ it is $\PSPACE$-complete \cite{DBLP:conf/tark/AucherS13}. Even though $\DEL$ and $\POL$ both involve model updates and show similar complexity results with respect to model checking, the proof methodologies involved, especially for the upper bound result, are entirely different. In fact, it is not at all obvious that the model checking problem for $\POL$ is decidable. The $\PSPACE$ algorithm that we construct for this problem and its proof of correctness %is involved and 
%crucially use Savitch's Theorem. %This upper bound result for the model checking problem for $\POL$ is the main novel contribution of this paper.

{\bf Dynamic epistemic reasoning} %We first recall that  
%that
The model checking of standard epistemic logic ($\EL$) is $\PTIME$-complete \cite{DBLP:conf/aiml/Schnoebelen02}.
Public Observation Logic ($\POL$) is quite similar to Public announcement logic ($\PAL$) \cite{DBLP:journals/synthese/Plaza07}. When public announcements are performed, the number of possible worlds reduces, making the model checking of $\PAL$ still in $\PTIME$ \cite{van2011logical} as for standard epistemic logic. When actions can be private, %the number of worlds may increase and 
the model checking becomes $\PSPACE$-complete for $\DEL$ with action models \cite{DBLP:conf/tark/AucherS13}.
%, and \cite{DBLP:conf/ijcai/BolanderJS15} and %\cite{DBLP:journals/corr/PolRS16} for a more refined lower bound) both without or with common knowledge \cite{DBLP:conf/aiml/CharrierS18}.

%Let us compare $\POL$ and $\PAL$ and discuss the impact on the complexity of the model checking. 
In $\PAL$, a possible world is equipped with a valuation, while in $\POL$ it is also equipped with a regular expression denoting the expectation in that world. In $\PAL$, the public announcement is fully specified and its effect is deterministic. In $\POL$, we may reason on sets of possible observations represented by regular expressions $\pi$. %Note that 
When these sets are singletons, we again obtain a $\PTIME$ upper bound (Theorem~\ref{theorem:word}). In this sense, $\POL$ is close to Arbitrary $\PAL$ ($\APAL$) \cite{DBLP:conf/aiml/FrenchD08} whose model checking is also $\PSPACE$-complete \cite{DBLP:journals/japll/AgotnesBDS10}. In $\APAL$, any epistemic formula can be announced: there are no expectations. However, in $\POL$, we have to reason about the constraints between the possible expectations, and the set of %possible 
observations (given by $\pi$).  Our contribution can be reformulated as follows: we prove that (i) reasoning about these constraints can still be done in $\PSPACE$, and, (ii) this reasoning is sufficiently involved for the model checking to be $\PSPACE$-hard. 

In $\POL$, regular expressions are used to represent sets of observations, %(finite sequences of action), 
while van Benthem et al. \cite{DBLP:journals/iandc/BenthemEK06} used regular expressions (actually, programs of Propositional dynamic logic ($\PDL$) \cite{DBLP:journals/jcss/FischerL79}) to denote epistemic relations. Charrier et al. \cite{DBLP:journals/logcom/CharrierPS19} considered a logic for reasoning about protocols where actions are public announcements and not abstract observations as in POL: in this sense, POL is more general.

%\todo{\cite{DBLP:conf/ijcai/BelardinelliKLM21}}

%\todo{compare with linear temporal epistemic logic because having (pointed) timelines as possible worlds is close to having knowledge about the course of protocols}

%\todo{ More recently, the notions of ``knowing whether,'' ``knowing why,'' and ``knowing how'' have also been investigated from a formal viewpoint (e.g., see \cite{wang2018}). }
%it was in the intro but it is better here

%\todo{$\POL$ as $\EL$ extended with \textit{Propositional dynamic logic (PDL)}}

%To verify certain claims about these scenarios gives rise to the study of the algorithmic and computational studies of problems like \textit{model checking} or \textit{satisfiability}, its decidability and complexity result. It is previously known that the model checking problem of Multi-Agent Epistemic logic is in P, which is trivially true from the fact that the model checking of Propositional Modal Logic is in P \cite{vardi1997modal}.\\
%In this paper, we study the complexity results regarding the model checking problem of Public Observation Logic. We prove that the model checking problem is PSPACE-Complete, and hence we give an efficient algorithm for the problem which runs in PSPACE. We also look at an interesting restriction of the problem where the algorithm turns out to run in P.

\vspace{.5em}

\noindent {\bf Epistemic temporal reasoning} It is natural to describe computational behaviours with regular expressions. Finite-state controllers, i.e., automata are used to describe policies in planning \cite{DBLP:conf/aips/BonetPG09}. Interestingly, Lomuscio and Michaliszyn \cite{DBLP:conf/kr/LomuscioM16} studied an epistemic logic where formulas are evaluated on intervals and the language provides Allen's operators on intervals: in their setting, the model is an interpreted system, and a propositional variable $p$ is true in an interval $I$ if the trace of $I$ matches a given regular expression associated to $p$. In contrast, $\POL$ is not based on an already set-up model but relies on updates in a model. Bozzelli et al. \cite{DBLP:conf/sefm/BozzelliMMP17} studied the complexity of the model checking of that logic depending on the restrictions on the allowed set of Allen's operators. Their framework is similar to ours because it relies on regular expressions but the approach is orthogonal to model updates and hence, to epistemic planning.

\vspace{.5em}

\noindent {\bf Epistemic planning}
As far as we know, epistemic planning frameworks (based on DEL \cite{DBLP:journals/ai/BolanderCPS20}, or the so-called MEP for Multi-agent Epistemic Planning \cite{DBLP:journals/ai/MuiseBFMMPS22}) all
provide a mechanism for reasoning about preconditions and effects of actions. Expectations about others or about the world are not dealt with. However,
Saffidine et al. \cite{DBLP:conf/aaai/SaffidineSZ18} propose a collaborative setup for epistemic planning where each agent executes its own knowledge-based policy/program (KBP) while agents commonly know all the KBPs that are being executed, meaning that agents expect that the other agents follow their own KBP. On the contrary, in $\POL$, observations are public but expectations are in general not commonly known. Reasoning about some epistemic properties that are true after the execution of any kind of KBPs is undecidable, but is $\PSPACE$-complete for star-free KBPs. The complexity is high for different reasons:  the initial model is represented symbolically; observations are not already public, and KBPs may contain tests. 

\vspace{.5em}

\noindent {\bf Strategic reasoning} Usually in logics for strategic reasoning (e.g., alternating-time temporal logic \cite{DBLP:journals/jacm/AlurHK02}, and strategy logic \cite{DBLP:journals/iandc/ChatterjeeHP10}), agents do not have expectations: an agent may consider all possible strategies for the others. Recently, Belnardinelli et al. \cite{DBLP:conf/ijcai/BelardinelliKLM21} propose a variant of strategy logic (SL) where a player may know completely the strategy of another player. In contrast, in $\POL$ agents may have partial information about the expectations. In $\POL$, agents also have higher-order knowledge about these expectations. In SL, strategies are abstract objects in the logical language whereas in $\POL$, observations are represented as composite structures that the agents can reason about, similar to the work on games and strategies presented in \cite{gr12}. % the observations they receive. 
In this sense, $\POL$ can be seen as $\EL$ extended with $\PDL$ operators. %similar to the dynamic logic and automata-theoretic studies on strategic reasoning, e.g., see \cite{gr12}.

\section{Conclusion}\label{concl}

In this paper, we showed that the model checking for $\POL$ is $\PSPACE$-complete. Such complexity studies were left open in \cite{van2014hidden}. We also identified more tractable fragments (see Figure~\ref{figure:results}) of $\POL$. %One natural question that comes to the fore is:
Finally, we discussed the applicability of our study in verifying various features of interactive systems related to epistemic planning. A discussion on implementation is also provided.

We leave the investigations on model checking for $\EPL$, an extension of $\POL$, also proposed in \cite{van2014hidden}, %As mentioned in the introduction, 
for future work. We also aim to study the satisfiability problems of $\POL$ and $\EPL$ %in future, 
by adapting  the techniques from \cite{DBLP:conf/tark/AucherS13,DBLP:conf/atal/Lutz06}.

Many interesting features of such interactive systems remain to be investigated : 
private observations, like in DEL with action models \cite{DEL}; dynamic aspects (e.g., changing expectations); richer languages of expectations (e.g., context-free grammars for expectations), among others. Symbolic model checking can be considered as well following the trends of \cite{DBLP:journals/logcom/BenthemEGS18} and  \cite{DBLP:journals/logcom/CharrierPS19}.

This paper also opens up a research avenue for developing variants and extensions for reasoning about expectations and observations that can be expressive enough with reasonable complexities for the model checking problem.

To sum up, POL mixes epistemic logic and language theory for modelling mechanisms of social intelligent agents, and the current investigations on model checking set it up as a useful tool in building social software for AI. 

\bibliographystyle{unsrt}
\bibliography{bibliography}

\newpage 
\appendix

\begin{center}
   \Huge{\textbf{Appendix}}
\end{center}

\ 

% \section{Definitions}
% \sujata{In the traffic example, the observation expression $(g^\ast ar^\ast)^\ast$ models the traveller's expectation of traffic signals in case she is in France. In the other one, the observation expression $m$ models the expectation of the receiver in case a decision is made. The size of these observation expressions is defined as follows.} 

%\input{Supplementary/fragmentsdefn}
\section{$\POL$ Model Checking is in $\PSPACE$}

In this section, we prove that $\POL$ model-checking is in $\PSPACE$. We prove it by showing that the algorithm  $\DecidePSPACE$ (presented in Section~\ref{results}), takes as  input a $\POL$ model $\M = \ldiamondarg{S,\sim,V,\Exp}$, an initial starting world  $s\in S$, and a $\POL$ formula $\varphi$ and returns $\T$ if and only if $\M, s \models \phi$, and at the same time, the algorithm  $\DecidePSPACE$ runs in polynomial space.

A crucial step in $\DecidePSPACE$ is when (in line~\ref{ln:oracle})  it 
uses an oracle to check if $\M' = \M|_w$ for some word $w \in \LL(\pi)$. 
So to prove that the algorithm $\DecidePSPACE$ we need to prove the existence of a polynomial space subroutine to check if $\M' = \M|_w$ for some word $w \in \LL(\pi)$. For this we the Algorithm~\ref{algoMAIN:algofororacle} which is a non-deterministic polynomial space algorithm to check the same. And using Savitch's theorem we can conclude that a deterministic polynomial space algorithm must also exist. 
We start by proving the correctness and complexity of Algorithm~\ref{algoMAIN:algofororacle} (in Section~\ref{sec:proofnd}) and then using this we present the proof of correctness and complexity of $\DecidePSPACE$ in Section~\ref{sec:proofmain}.

%provide the correctness of the algorithm $\DecidePSPACE$ and also prove that it run in polynomial space.

\subsection{Correctness and Complexity of Algorithm~\ref{algoMAIN:algofororacle}}\label{sec:proofnd}

Before we prove the correctness of $\DecidePSPACE$ we need to first prove that there is a PSPACE algorithm (oracle) for checking if $\M' = \M|_w$ for some word $w \in \LL(\pi)$. We will first prove that Algorithm~\ref{algoMAIN:algofororacle} is a non-deterministic algorithm for  checking if $\M' = \M|_w$ for some word $w \in \LL(\pi)$ that takes polynomial space. For that we need to understand the following: for these two models $\M$ and $\M'$, what is the length of the smallest $w\in \LL(\pi)$ such that $\M' = \M|_w$? %We will show  in 
Lemma~\ref{theorem:ExponentialWitness} takes care of this query. We start with Observation~\ref{observation:prefixstatequivalence} which is required to show that the length of the smallest such $w$ is bounded by $2^{\pi}\times \Pi_{t \in \M} 2^{|\Exp(t)|}$.

%Once we have that observation, the correctness of the Algorithm~\ref{algoMAIN:algofororacle} follows, as we see below.

\begin{restatable}{observation}{prefixstatequivalence}
For a finite model $\M = \langle S,\sim,V,\Exp\rangle$, a world $s\in S$, and for every pair of words $w, w'\in \Sigma^*$, if $w$ and $w'$ are simulated in the $DFA(\Exp(s)) = (Q_s, \Sigma, \delta_s, q^0_s, F_s)$, with both simulations ending in the state $q\in Q_s$, then $DFA(\Exp(s)\backslash w) = DFA(\Exp(s)\backslash w')$.
\label{observation:prefixstatequivalence}
\end{restatable}

\begin{proof}
%By assumption $\widehat{\delta_s}(q^0_s, w) = \widehat{\delta_s}(q^0_s, w')$. 
From the definition of residue (Definition~\ref{def:residue}) we have 
$u\in \mathcal{L}(DFA(\Exp(s)\backslash w))$ iff $wu\in \mathcal{L}(DFA(\Exp(s)))$.  Note that,
% \begin{align*}
%     & wu\in \mathcal{L}(DFA(\Exp(s))) \hspace{12.0cm}\\
%   \mbox{ \textit{iff} } &u\in \mathcal{L}(DFA(\Exp(s)\backslash w))\\ 
%   \mbox{ \textit{iff} } & wu\in \mathcal{L}(DFA(\Exp(s)))\\
%     \mbox{ \textit{iff} } & \widehat{\delta_s}(q^0_s, wu)\in F_s\\
%     \mbox{ \textit{iff} } & \widehat{\delta_s}(\widehat{\delta_s}(q^0_s, w), u)\in F_s\\
%     \mbox{ \textit{iff} } & \widehat{\delta_s}(\widehat{\delta_s}(q^0_s, w'), u)\in F_s,
% \end{align*}
\begin{align*}
     wu\in \mathcal{L}(DFA(\Exp(s)))&\mbox{ \textit{iff} } u\in \mathcal{L}(DFA(\Exp(s)\backslash w)) \hspace{12.0cm}\\
   &\mbox{ \textit{iff} }  wu\in \mathcal{L}(DFA(\Exp(s)))\\
    &\mbox{ \textit{iff} } \widehat{\delta_s}(q^0_s, wu)\in F_s\\
    &\mbox{ \textit{iff} }  \widehat{\delta_s}(\widehat{\delta_s}(q^0_s, w), u)\in F_s\\
    &\mbox{ \textit{iff} }  \widehat{\delta_s}(\widehat{\delta_s}(q^0_s, w'), u)\in F_s,
\end{align*}
the last if and only if holds as by assumption $\widehat{\delta_s}(q^0_s, w) = \widehat{\delta_s}(q^0_s, w')$.  Finally note that, 

\begin{align*}
     \widehat{\delta_s}(\widehat{\delta_s}(q^0_s, w'), u)\in F_s 
    \mbox{ \textit{iff} } &  w'u\in \mathcal{L}(DFA(\Exp(s)))\\
    \mbox{ \textit{iff} } & u\in \mathcal{L}(DFA(\Exp(s)\backslash w'))
\end{align*}

\end{proof}
% \textcolor{red}{
% \begin{proof}
% By assumption $\widehat{\delta_s}(q^0_s, w) = \widehat{\delta_s}(q^0_s, w') = q$. 
% From the definition of residue (Definition~\ref{def:residue}) we have 
% $u\in \mathcal{L}(DFA(\Exp(s)\backslash w))$ if and only if $wu\in \mathcal{L}(DFA(\Exp(s)))$. 

% Hence
% \begin{align*}
%     &u\in \mathcal{L}(DFA(\Exp(s)\backslash w)) \mbox{ \textit{iff} } wu\in \mathcal{L}(DFA(\Exp(s)))\\
%     &\mbox{ \textit{iff} } \widehat{\delta_s}(q^0_s, wu)\in F_s
%     \mbox{ \textit{iff} } \widehat{\delta_s}(\widehat{\delta_s}(q^0_s, w), u)\in F_s\\
%     &\mbox{ \textit{iff} } \widehat{\delta_s}(\widehat{\delta_s}(q^0_s, w'), u)\in F_s
%     \mbox{ \textit{iff} } w'u\in \mathcal{L}(DFA(\Exp(s)))\\
%     &\mbox{ \textit{iff} } u\in \mathcal{L}(DFA(\Exp(s)\backslash w'))
% \end{align*}
% \end{proof}
% }

Using Observation~\ref{observation:prefixstatequivalence} we can obtain an upper bound on the size of the set $\Gamma^\M = \{\M|_w\mid w\in\Sigma^*\}$.

\begin{restatable}{lemma}{NumberOfProjectedModels}\label{lem:firstnd}
Given a finite $\POL$ model $\M = \langle S,\sim,V,\Exp \rangle$, the size of $|\Gamma^\M|\leq\Pi_{t \in \M} 2^{|\Exp(t)|}$.
\label{lemma:NumberOfProjectedModels}
\end{restatable}
\begin{proof}
For any given world $s\in S$ and $DFA(\Exp(s)) = (Q_s, \Sigma, \delta_s, q^0_s, F_s)$, relation $Z^\M_s\subseteq \Sigma^* \times \Sigma^*$ is defined as:
$$
(w,u)\in Z^\M_s\mbox{ iff } \widehat{\delta_s}(q^0_s, w) = \widehat{\delta_s}(q^0_s, u)
$$
% <<<<<<< HEAD
Clearly, $Z^\M_s$ is an equivalence relation, hence creates a partition over $\Sigma^*$. Therefore by Observation \ref{observation:prefixstatequivalence}, for any pair $(w,w')\in Z^\M_s$, $DFA(\Exp(s)\backslash w) = DFA(\Exp(s)\backslash w')$. In other words, any $w$ from a single partition $\big[[w]\big]_s$ over $\Sigma^*$ by $Z^\M_s$, will produce the same $DFA(Exp(s)\backslash w)$.
Therefore, number of partitions over $\Sigma^*$ by $Z^\M_s$ is at most the number of states in $DFA(\Exp(s))$, that is, $2^{|\Exp(s)|}$.

For the model $\M$, let $n = |S|$. Consider the following $n$-tuple $T_\M^w = (D^w_{s_1},\dots,D^w_{s_n})$ for all $w\in\Sigma^*$, where $D^w_{s_i} = DFA(Exp(s_i)\backslash w)$ for the world $s_i\in S$. Note that $|\Gamma^\M| = |\{T^w_\M \mid w\in\Sigma^* \}|$, because for every $w\in \Sigma^*$, $T^w_\M$ is the tuple enumerating the $Exp$ function of $\M|_w$ according to the worlds of $\M$ (Note that, if a world vanishes in $\M|_w$ for some $w$, the corresponding DFA will be of empty language). 

% =======
% Clearly, $Z^\M_s$ is an equivalence relation, hence creates a partition over $\Sigma^*$. Therefore by Observation \ref{observation:prefixstatequivalence}, for any pair $(w,w')\in Z^\M_s$, $DFA(\Exp(s)\backslash w) = DFA(\Exp(s)\backslash w')$. In other words, any $w$ from a single partition $\big[[w]\big]_s$ over $\Sigma^*$ by $Z^\M_s$, will produce the same $DFA(\Exp(s)\backslash w)$.\\
% Therefore, number of partitions over $\Sigma^*$ by $Z^\M_s$ is at most the number of states in $DFA(\Exp(s))$, that is, $2^{|\Exp(s)|}$.\\
% For the model $\M$, let $n = |S|$. Consider the following $n$-tuple $T_\M^w = (D^w_{s_1},\dots,D^w_{s_n})$ for all $w\in\Sigma^*$, where $D^w_{s_i} = DFA(\Exp(s_i)\backslash w)$ for the world $s_i\in S$. Note that $|\Gamma^\M| = |\{T^w_\M \mid w\in\Sigma^* \}|$, because for every $w\in \Sigma^*$, $T^w_\M$ is the tuple enumerating the $\Exp$ function of $\M|_w$ according to the worlds of $\M$ (Note that, if a world vanishes in $\M|_w$ for some $w$, the corresponding DFA will be of empty language). \\
% >>>>>>> 7cc15e7d7b6cd57636a58f378c4cb0d5b3ff7369
For each world $s_i$, the total number of $D^w_{s_i}$ possible is at most $2^{|\Exp(s_i)|}$, and hence the total number of such tuples possible is $\Pi_{t \in \M} 2^{|\Exp(t)|}$.
\end{proof}

Now using the Lemma~\ref{lem:firstnd} we prove Lemma~\ref{theorem:ExponentialWitness} that would be used to prove the correctness of the Algorithm~\ref{algoMAIN:algofororacle}.

\begin{lemma}
Given a $\POL$ model $\M=\langle S,\sim,V,\Exp\rangle$, a world $s\in S$ and a formula $\ldiamondarg{\pi}\psi$, $\M,s\vDash\ldiamondarg{\pi}\psi$ iff $\exists w\in \LL(\pi)$ of length at most $2^{|\pi|}\times \Pi_{t \in \M} 2^{|\Exp(t)|}$ such that $\M|_w,s\vDash\psi$ and the world $s$ survives in $\M_w$.
%\label{lemma:singleEXPWitness}
\label{theorem:ExponentialWitness} 

%For any $\POL$ formula $\varphi$, a finite $\POL$ model $\M$ and a world $s$ there exists a $w\in\Sigma^*$, if $\M|_w,s\vDash\varphi$ then $|w|\leq |Q_\M|^{|S|}$.
\end{lemma}
\begin{proof}
The $\Leftarrow$ direction is easy: if there exists a $w\in \LL(\pi)$ such that $\M|_w,s\vDash\psi$ then by definition $\M, s\vDash \ldiamondarg{\pi} \psi$. 

Now for the $\Rightarrow$ direction, consider the edge graph $G^\M(\Gamma^\M, E^\M)$ on vertex set $\Gamma^{\M}$ and and edge from vertex $\M|_u$  to vertex $\M|_{u'}$ is present if and only if there exists a $a\in \Sigma$ such that $\M|_{ua} = \M|_{u'}$. From Lemma~\ref{lemma:NumberOfProjectedModels} we know that the number of vertices in the graph is at most $\Pi_{t \in \M} 2^{|\Exp(t)|}$. %$|Q_m|^{O|S|}$.
Thus
it is easy to observe that for any two vertices $\M|_u, \M|_{u'} \in \Gamma^{\M}$ 
the set $\Delta_{\M|_u, \M|_{u'}}:= \{w\in \Sigma^* \mid \M|_{uw} = \M|_{u'}$\} is a regular language accepted by a DFA of size at most $\Pi_{t \in \M} 2^{|\Exp(t)|}$. 

Let $\M,s\vDash\ldiamondarg{\pi}\psi$. Then by definition there exists a $w_0\in \LL(\pi)$ such that $\M|_{w_0}, s\vDash \psi$. Note that $\M$ and $\M|_{w_0}$ are both vertices of the graph $G^{\M}$. So all the $\{w\in \LL(\pi) \mid \M|_{w} = \M|_{w_0}\}$ is nothing but the the set $\LL(\pi) \cap \Delta_{\M, \M|_{w_0}}$. So we know there is a $w \in \LL(\pi) \cap \Delta_{\M, \M|_{w_0}}$ of size at most $2^{\pi}\times \Pi_{t \in \M} 2^{|\Exp(t)|}$ and for that $w$, 
$\M|_w, s\vDash \psi$. 
% Assume that $\M|_w,s\vDash\varphi$. Create a graph structure $G^\M(\Gamma^\M, E^\M)$ with labelled edges such that $(\M|_u, \M|_{u'}, x)\in E^\M$ for all $x\in\Sigma$ iff $ux\in \big[[u']\big]_s\ \ \forall s\in S$, where $S$ is the set of worlds of $\M|_{u'}$ and $\big[[u']\big]_s$ is the equivalent class containing $u'$ as constructed in proof of Lemma \ref{lemma:NumberOfProjectedModels}. \\
% Now consider a path in $G^\M$ from $\M$ (that is, $\M|_\epsilon$) to $\M|_w$. Each edge in path contributes an alphabet $x$ to the string the adjascent model is projected on and the biggest path that can occur in this graph is $|\Gamma^\M|$. Hence there exists a $w$ such that $|w|\leq |Q_\M|^{|S|}$.
\end{proof}

From Lemma~\ref{theorem:ExponentialWitness} we know that if $\M' = \M|_w$ for some $w\in \LL(\pi)$ there exists a $w_0\in \LL(\pi)$ with $|w_0|\leq 2^{|\pi|}\times \Pi_{t \in S} 2^{|\Exp(t)|}$ and $\M' = \M|_{w_0}$. Since Algorithm~\ref{algoMAIN:algofororacle} guesses a $w = \alpha_1 \dots, \alpha_j \dots$ (one letter at a time) of length at most $2^{|\pi|}\times \Pi_{t \in S} 2^{|\Exp(t)|}$ and $\M' = \M|_{w_0}$ and checks if $\M'=\M|_w$, from Lemma~\ref{theorem:ExponentialWitness} we see that the algorithm is correct. Note that Algorithm~\ref{algoMAIN:algofororacle} is a non-deterministic algorithm. The algorithm uses only polynomial space (in the size of $\M$), since at any point of time (say at the $j$th iteration of the \textbf{for} loop in Line~\ref{ln:oracleforstart}) the algorithm only have to update the model from $\M|_{\alpha_1 \dots \alpha_{j-1}}$ to  $\M|_{\alpha_1 \dots \alpha_{j}}$ which can be done using polynomial space. Note that the Algorithm~\ref{algoMAIN:algofororacle} does not have to remember the string $\alpha_1 \dots \alpha_j\dots $ which can be of size exponential.  So the Algorithm~\ref{algoMAIN:algofororacle} is a non-deterministic polynomial space algorithm.  Thus we have

\begin{theorem}\label{thm:nd}
Algorithm~\ref{algoMAIN:algofororacle} is a non-deterministic polynomial space algorithm that correctly checks if  $\M' = \M|_w$ for some $w\in \LL(\pi)$. 
\end{theorem}

By Savitch's Theorem \cite{DBLP:journals/jcss/Savitch70},% this implies tha
there is a deterministic polynomial space oracle for checking if $\M' = \M|_w$ for some $w\in \LL(\pi)$.  Thus we have 

\begin{theorem}\label{thm:det}
There is a deterministic polynomial space algorithm that correctly checks if  $\M' = \M|_w$ for some $w\in \LL(\pi)$. 
\end{theorem}

\subsection{Correctness and Complexity of $\DecidePSPACE$}\label{sec:proofmain}

Using the Theorem~\ref{thm:det} we now present the proof of correctness and complexity of $\DecidePSPACE$.

\begin{restatable}{lemma}{correctness}\label{thm:correctness}
$\DecidePSPACE(\M,s,\varphi)$ returns $\T$ iff $\M,s\vDash\varphi$.
\end{restatable}
\begin{proof}
We will prove $\DecidePSPACE(\M, s, \varphi)$ returns $\T$ iff $\M,s\vDash\varphi$ by induction on the size of $\varphi$.
%\textbf{Base Case. }Consider $\varphi$ to be propositional formula. 
%$\M,s\vDash\varphi$ iff the propositional variables are assigned according to $V(s)$ iff $\DecidePSPACE(\M, s, \varphi)$ returns TRUE.

\noindent \textbf{Base Case. }Consider $\varphi$ to be a proposition. 
$\M,s\vDash\varphi$ iff $\varphi\in V(s)$ iff $\DecidePSPACE(\M, s, \varphi)$ returns $\T$.

\noindent\textbf{Induction Hypothesis. } For any $\POL$ formula $|\psi|\leq m$, any finite model $\M$ and any world $s$, $\DecidePSPACE(\M, s, \psi)$ returns $\T$ iff $\M,s\vDash\psi$.

\noindent\textbf{Inductive Step. } We go case by case over the forms of $\varphi$. For all the cases except when $\varphi = \ldiamondarg{\pi}\psi$, the inductive step is trivial. So we focus on the crucial case when $\varphi = \ldiamondarg{\pi}\psi$.

By definition we know that $\M, s\vDash \ldiaarg \pi \psi$ iff there exists a $w\in \LL(\pi)$ such that $\M|_w, s\vDash \psi$.  In other words, $\M, s\vDash \ldiaarg \pi \psi$ iff there exists $\M'\in \Gamma^{\M}$ such that $\M' = \M|_w$ for some $w\in \LL(\pi)$ and the world $s$ survives and $\M', s\vDash \psi$. By induction hypothesis $\M', s\vDash \psi$ iff $\DecidePSPACE(\M',s, \psi)$ is $\T$.  In the \textbf{for} loop Lines~\ref{line:forallmodelstart} to \ref{line:forallmodelend} the algorithm goes over all $\M'\in \Gamma^{\M}$.  For each of the $\M'$ the algorithm in Line~\ref{line:oraclecall} calls the oracle (Algorithm~\ref{algoMAIN:algofororacle}) to check if $\M' = \M|_w$ for some $w\in \LL(\pi)$ and if the world survives calls $\DecidePSPACE(\M',s, \psi)$ recursively. Since we have already argued correctness of Algorithm~\ref{algoMAIN:algofororacle} by Theorem~\ref{thm:nd} , the correctness of the algorithm follows. 
\end{proof}

Now we move on to prove that $\DecidePSPACE$ uses polynomial amount of space. This (along with Lemma~\ref{thm:correctness} and Theorem~\ref{thm:det}) would prove that $\DecidePSPACE$ is in $\PSPACE$.

\begin{restatable}{lemma}{mcspace}\label{thm:space}
	$\DecidePSPACE$ uses polynomial space.
% 	\label{theorem:full}
\end{restatable}
\begin{proof}
Since the algorithm is recursive, the formal argument (as in the proof of Lemma~\ref{thm:correctness}) should go via induction. It 
can be observed that in all the cases except when $\varphi = \ldiamondarg{\pi}\psi$, $\DecidePSPACE$ only uses a constant amount of space before making the recursive call. In the case when $\varphi = \ldiamondarg{\pi}\psi$, since the algorithm goes over all $\M'\in \Gamma^{\M}$(\textbf{for} loop from Line~\ref{line:forallmodelstart} to Line~\ref{line:forallmodelend}), the algorithm  will have to do some bookkeeping to keep a track on when $\M$ is being processed and to store the current $\M$. But since $\Gamma^{\M}$ has size exponential
(Lemma~\ref{lemma:NumberOfProjectedModels}) and since all $\M'\in \Gamma^{\M}$ can be represented in size polynomial in the size of $\M$ so it is possible to do the bookkeeping and tracking using only polynomial space. For any $\M'$ inside the \textbf{for} loop the only non-trivial thing to do is the call to the oracle in Line~\ref{line:oraclecall}. By Theorem~\ref{thm:det}, there exists an algorithm in $\PSPACE$ that given two models $\M$ and $\M'$ and regular expression $\pi$ checks if there exists $w\in \LL(\pi)$ such that $\M' = \M|_w$. This space can of course be reused for any iteration of the \textbf{for} loop in Line~\ref{line:forallmodelstart}. So $\DecidePSPACE$ uses at most polynomial space before making a recursive call and hence the total space used by the algorithm is polynomial. 
% 	We prove that the algorithm $\DecidePSPACE(\M,s,\varphi)$ runs in polynomial space.
% We will prove this by induction on size of formula $|\varphi| = n$. Let the space taken by the algorithm be $S(n)$. In the case of $\varphi = \ldiamondarg{\pi}\psi$, a number of value at most $2^{|\pi|}\times \Pi_{t\in S}2^{|\Exp(t)|}$ takes at most $|\pi| + \Sigma_{t\in S}|\Exp(t)|$ bits. As it can be seen this algorithm is a non-deterministic one. Also each NFA is updated with constant space (Only the start states are being changed by one letter). Hence $S(n)\leq S(n-1) + O(n^c)$, where $c$ is a constant, since after non-deterministically guessing the correct string, the algorithm is recursively called with $|\psi| = n - 1$.\\
% \textbf{Base Case. }$n=1$. hence $\varphi$ is a propositional letter. Determining whether the letter is in $V(s)$ requires constant amount of space. Hence $S(n)\leq d\leq O(n^c)$, where $d$ is a constant.\\
% \textbf{Induction Hypothesis. }For all $i\leq n$, $S(i)\leq O(i^c)$, where $c$ is a constant.\\
% \textbf{Inductive Step. }
% \begin{align*}
%     S(n + 1)&\leq S(n) + O((n+1)^c)\\
%     &\leq O(n^{c'}) + O((n+1)^c)
%     \leq O((n+1)^c),
% \end{align*}
% where the second inequality follows from IH.
\end{proof}

By combining Lemma~\ref{thm:correctness} and \ref{thm:space} we have the following:

\begin{theorem}
$\POL$ model-checking is in $\PSPACE$.
\end{theorem}
\section{Model checking for $\POL$ is $\PSPACE$-hard}

As pointed out in Section~\ref{results} there are two sources for the model checking to be $\PSPACE$-hard: Kleene star in observation modalities as well as alternations in modalities (sequences of nested existential and universal modalities). 
%In Theorem~\ref{theorem:existential} and Theorem~\ref{theorem:starfree} 
We prove the $\PSPACE$-hardness of the model-checking against the $\existential$ fragment  and the $\starfree$ fragment of $\POL$ respectively.

\existentialLB*
\begin{proof}
%https://en.wikipedia.org/wiki/Intersection_Non-Emptiness_Problem#cite_note-Kozen1977-1
% \newcommand{\automaton}{\mathcal A}
% \newcommand{\modelM}{\mathcal M}
% \newcommand{\languageof}[1]{L({#1})}
% \newcommand{\set}[1]{\{#1\}}
% \newcommand{\suchthat}{\mid}
\cite{DBLP:conf/focs/Kozen77} proved that the following problem,  called the intersection non-emptiness problem, is $\PSPACE$-complete: given a finite collection of DFAs $\automaton_1, \dots, \automaton_n$, decide whether $\LL(\automaton_1) \cap \dots \cap \LL(\automaton_n)  \neq \emptyset$.  Let us reduce this problem to the model checking for $\POL$. For that we construct the instance $\tr(\automaton_1, \dots, \automaton_n) = (\modelM, s_0, \phi)$ as:
% In \cite{DBLP:conf/focs/Kozen77}, it is proven that the following problem is $\PSPACE$-complete: given a finite collection of DFAs $\automaton_1, \dots, \automaton_n$, decide whether $\languageof{\automaton_1} \cap \dots \cap \languageof{\automaton_n}  \neq \emptyset$. This problem is called the intersection non-emptiness problem. Let us reduce it to the model checking against $\POL$. For that we construct the instance $tr(\automaton_1, \dots, \automaton_n) = (\modelM, s_0, \phi)$ defined by:
% %\todo{$R$ should be $\sim$}
\begin{itemize}
    \item $\M = (S, \sim, V, Exp)$ where 
    \begin{itemize}
    	\item     $S = \set{0,1, \dots, n} \cup \set{0', 1', \dots, (n-1)'}$.
    	\item $\sim_1 = \set{(i, i'), (i, i), (i', i'), (i', i)  \suchthat i=0, \dots, n-1}$,
    	$\sim_2 = \set{(i', i+1), (i', i'), (i+1, i+1), (i+1, i') \suchthat i=0, \dots, n-1}$,
    	\item $V(s) = V(s') = \emptyset$ for all $s \leq n-1$, $V(n) = \set{p}$.
    	\item     $Exp(0) = Exp(0') = \Sigma^*$ and $Exp(i) = Exp(i') = \automaton_i$ for all $i=1\dots n$.
    \end{itemize}
    \item $s_0 = 0$.
    \item $\phi = \ldiamondarg{\Sigma^*} (\hat K_1\hat K_2)^{n+1} p$.
\end{itemize}

The model $\M$ is a chain of worlds: starting with two worlds $0$, $0'$, labeled by an automaton for the universal language $\Sigma^*$. followed by worlds $1$, $1'$ labelled by automaton $\automaton_1$, followed by worlds $2$, $2'$ labelled by automaton $\automaton_2$, etc. It ends with worlds $n-1$, $(n-1)'$ labelled by automaton $\automaton_{n-1}$, followed by a world $n$ labelled by $\automaton_{n}$. Proposition $p$ is false in all worlds except $n$. The formula $\phi$ says that there exists a word $w$ such that $ (\hat K_1\hat K_2)^{n+1} p$ holds in $\M|_w, 0$. As $n$ should still be reachable, it means that the word must be in $\LL(\automaton_i)$ for all $i\in [n]$.
Now, $\tr(\automaton_1, \dots, \automaton_n)$ is computable in polynomial time in the size of
$(\automaton_1, \dots, \automaton_n)$.
Furthermore we have $\LL(\automaton_1) \cap \dots \cap \LL(\automaton_n)  \neq \emptyset$ iff $\modelM , s_0 \models \phi$.
\end{proof}

\ 

\starfreeLB*
\begin{proof}
    We shall prove by a reduction from TQBF. % problem.
    Given a Quantified Boolean formula $\varphi = Q_1x_1\dots Q_nx_n \gamma$, where $\gamma$ is in CNF containing $n$ variables $\{x_1,\dots ,x_n \}$ and $m$ clauses $C = \{c_1,\dots ,c_m \}$ and $Q_i\in\{\exists, \forall\}$ for all $i\in [n]$,  we shall construct our reduction. Consider the translations $\tr(x_i) = a_i$, $\tr(\neg{x_i}) = a'_i$, $\tr(\exists x_i) = \ldiamondarg{a_i + a'_i}$ and $\tr(\forall x_i) = [a_i + a'_i]$ for all $i\in [n]$. Now we present the model checking instance that we construct from $\varphi$.
    \begin{itemize}
        \item  Alphabet: $\Sigma_\varphi = \cup_{i\in [n]}\{\tr(x_i), \tr(\neg{ x_i})\}$.
        %corresponding to positive and negative literal of the variable $x_i$ respectively.
        \item Model: $\M_\varphi = \langle S_\varphi,\sim_\varphi,V_\varphi,Exp_\varphi \rangle$, where
        \begin{itemize}
            \item $S_\varphi = \{1,\dots, m \}$.
            \item A single agent $1$ and $\forall i, j\in S_{\varphi}$, $i\sim_1 j$. %\todo{correct the indices: $i$ is for indexing the propositions, use $j$ instead}
            \item $V_\varphi(j) = \{p_j\}$ for all $j\in S_{\varphi}$. %\todo{use $j$ instead}
            \item For each clause $c_j\in C$, $Exp_\varphi(j) = (\sum_{\ell\in c_j}(\Sigma_{\varphi}\setminus\{\tr(\ell), \tr(\neg{\ell})\})^*\tr(\ell)(\Sigma_{\varphi}\setminus\{\tr(\ell), \tr(\neg{\ell})\})^*)$, where the sum is over the literals in the clause $c_j$. 
            % For example, for a clause $c_j = (x_3\vee \neg{x_4}\vee x_2)$ in $\varphi(x_1,..,x_n)$, $Exp_\varphi(j) = ((\Sigma\setminus\{a_3, a'_3\})^*a_3(\Sigma\setminus\{a_3,a'_3\})^* + (\Sigma\setminus\{a_4, a'_4\})^*a'_4(\Sigma\setminus\{a'_4, a_4\})^* + (\Sigma\setminus\{a_2,a'_2\})^*a_2(\Sigma\setminus\{a_2, a'_2\})^*)$.
            
            %\todo{add a remark that the star in Exp could also be removed no? simply repeat the pattern $n$ times (see the formula)}
        \end{itemize}
        \item Formula: $\psi := \tr(Q_1x_1)\dots \tr(Q_nx_n)\bigwedge_{i\in [m]}(\hat{K_1}p_i)$.
        \item Starting world: $s$ is any world from $S$.
    \end{itemize}
    
    We need to prove that the QBF $\varphi$ is $\T$ iff $\M_\varphi, s\vDash \psi$.
        Let us start by proving the forward direction: that is if $\varphi$ is $\T$ then we prove that $\M_{\varphi}, s \vDash \psi$. 
    
    Consider the set of $\mathcal{T}$ all assignments of $(\ell_1, \dots, \ell_n)$ (with $\ell_i \in \{x_i, \neg{x_i}\}$ that make the CNF formula $\gamma$ evaluates to $\T$. Since we assumed that the QBF $\varphi$ is $\T$ we observe that there exist a subset $\mathcal{T}'\subseteq \mathcal{T}$ that has the ``structure" of $Q_1x_1\dots Q_nx_n$. By the construction of $\tr$ and the formula $\psi$ we see that all we need to show is that for all assignment $(\ell_1, \dots, \ell_n)\in  \mathcal{T}'$ $\M_{\varphi}|_{\tr(\ell_1)\dots \tr(\ell_n)}, s\vDash \psi$. 
    
    Let $w = \tr(\ell_1)\dots \tr(\ell_n)$ where $(\ell_1, \dots, \ell_n) \in \mathcal{T}'$. 
    Consider any world $j$ (corresponding to the clause $c_j = (\ell_p\vee \ell_q\vee \ell_r)$). What happens to the world $j$ in $\M_{\varphi}|_w$?
    If we think of $\M_{\varphi}|_w$ as a series of $n$ updates (namely, $\M_{\varphi}|_{\tr(y_1)}, \M_{\varphi}|_{\tr(y_1)\tr(y_2)} \dots \M_{\varphi}|_{\tr(y_1)\dots  \tr(y_n)}$) then note that if the variable $x_i$
   is not in the clause $c_j$ then the update by $\tr(y_i)$ does not affect the world $j$. At the same time, since $y_1\dots y_n$ is a satisfying assignment to 
   the $\varphi$ so at least one of $\ell_p, \ell_q, \ell_r$ is in the set $\{y_1, \dots, y_n\}$ and hence after the updating by $\tr(y_p), \tr(y_q)$ and $\tr(y_r)$
   the world $j$ survives. Wlog if we assume $\ell_p = y_p$ then the $Exp(j)$ will have $(\Sigma\setminus\{\tr(\ell_p),\tr(\neg{\ell_p})\})^*$ added in the updated model, which guarantees the survival of the world in subsequent updates.
   
   Thus for any world $j$ (corresponding to the clause $c_j = (\ell_p\vee \ell_q\vee \ell_r)$), in the $\M_{\varphi}|_w$ the world $j$ survives, that is, no world vanishes after the update. And since $s\sim_1 j$ for all $j$, hence, $\M_{\varphi}|_w, s\vDash \bigwedge_{i\in [m]}(\hat{K_1}p_i)$ for any $1\leq j\leq m$, since all the initial $m$ worlds are in the same equivalence class of $\sim_1$. Since this is true for any $w$ corresponding to any satisfying assignment of $\gamma$, so $\M_{\varphi}, s\vDash \psi$.
    
    Conversely, assume $\varphi$ is unsatisfiable. Note that if for a set of literals  $\ell := \ell_1, \dots, \ell_n$  if $\gamma$ evaluates to $\Fa$ then 
    there exists $i_{\ell}$ such that the world $i_{\ell}$ does not survive in $\M_{\varphi}|_{\tr(\ell_1)\dots \tr(\ell_n)}$. This is because 
    there exist at least a clause, say $c_{i_{\ell}}$ in $\gamma$, that evaluates to $\Fa$ when the the literals $\ell$ is assignment $\T$. Consider the 
    The world $i_{\ell}$ in $\M_\varphi$ corresponding to $c_{i_{\ell}}$, will have $Exp(i_{\ell})\backslash w = \delta$. In that case, note that $\M_\varphi|_w, s \nvDash \hat{K_1}p_{i_{\ell}}$.  Thus there is a bijection between the set $\mathcal{T}$ of all satisfying assignments of $\gamma$ and the set $\{w\mid \M_\varphi|_w, s \vDash \hat{K_1}p_{i_{\ell}} \}$. Now from the construction of the formula $\psi$ we see that $\varphi$ is satisfiable iff $\M, s\vDash \psi$.
\end{proof}

\section{Complexity Results for Model checking for $\starfree-\existential$ and
$\word$ fragment of $\POL$}

In this section we prove the complexity of model-checking in the $\starfree-\existential$ and $\word$ fragment of $\POL$. 

\subsection{Complexity for Model checking for $\starfree-\existential$ fragment of $\POL$}

We start with proving that the model-checking problem for the $\starfree-\existential$ is in $\NP$. To prove that the model-checking problem for the $\starfree-\existential$ is in $\NP$ we present an algorithm $\GetSetNP$ and in Lemma~\ref{thm:NP} we prove the correctness and complexity of the algorithm $\mcNP$. 

In Lemma~\ref{thm:NPH} we prove that the model-checking problem in $\NP$-hard. Thus Lemma~\ref{thm:NP} and \ref{thm:NPH} combined gives us 

\starfreeexistNPC*

The algorithm $\mcNP$ calls a non-deterministic subroutine $\GetSetNP$ that takes the model $\M$ and NNF formula $\phi$, which is a $\starfree-\existential$ formula, and returns the set of all worlds $s$ such that $\M,s\vDash \phi$. The subroutine $\GetSetNP$ uses another subroutine $\ResidueByLetter$ that in turn uses a subroutine $\AuxOut$. The goal of the subroutine $\ResidueByLetter$ is to take as input a regular expression $\pi$ and a word $a$ and output the residue $\pi\regdiv a$ in polynomial time. Although the algorithm $\ResidueByLetter$  is straightforward for the sake of completeness we present the pseudo-code formally. The proof of correctness and the complexity of the algorithms (formally stated in Lemma~\ref{lemma:residuebyletter}) follows from standard arguments.

\begin{algorithm}[t]
%\color{red}
\caption{$\mcNP$}
\textbf{Input}: $\M = \ldiamondarg{S, R, V, Exp}$, $s\in S$, $\varphi$, where $Exp$ are NFA, and $\varphi$ is a $\starfree-\existential$ Formula\\
\textbf{Output}: Returns $\T$ iff $\M,s\vDash\varphi$
\begin{algorithmic}[1] %[1] enables line numbers
\IF{$s\in \GetSetNP(\M, \phi)$}
\STATE Return $\T$
\ENDIF
%\ELSE{Return $\F$}
\end{algorithmic}
\label{algo:mcSFSTAR}
\end{algorithm}
\vspace{-0.5cm}
\begin{algorithm}[t]
%\color{red}
\caption{{$\GetSetNP$}}
\textbf{Input}: $\M = \ldiamondarg{S, R, V, Exp}, \varphi$, where $Exp$ are NFA, and $\varphi$ is a $\starfree-\existential$ Formula\\
\textbf{Output}: Returns set of states $S'\subseteq S$ such that $\M,s\vDash\varphi$ for all $s\in S'$
\begin{algorithmic}[1] %[1] enables line numbers
\IF{$\varphi=p$ is a propositional variable }\label{getsetnp:basecase}
    \STATE $S'=\emptyset$\;
    \FOR{$s\in S$}
        \IF{$p\in V(s)$}
            \STATE $S' = S'\cup \{ s\}$\;
        \ENDIF
    \ENDFOR
    \STATE Output $S'$\;
\ENDIF
\IF{$\varphi = \neg p$, where $p$ is a propositional variable}
	\STATE Output $S\setminus \GetSetNP(\M,\psi)$
\ENDIF
\IF{$\varphi = \psi_1 \vee \psi_2$}
    \STATE Output $\GetSetNP(\M,\psi_1)\cup\GetSetNP(\M,\psi_2)$
\ENDIF
\IF{$\varphi = \psi_1\wedge\psi_2$}
    \STATE Return $\GetSetNP(\M,\psi_1)\cap\GetSetNP(\M,\psi_2)$
\ENDIF
\IF{$\varphi = \ldiamondarg{\pi}\psi$}
    \STATE $\pi' = \pi$\\
    \STATE $\M' = \M$
    \FOR{ $i = 1$ to $|\pi|$}\label{ln:NPguessloop}
        \IF{$\epsilon\in\pi'$ and $s\in S$}\label{ln:piexhaust}
            \STATE return $\GetSetNP(\M',\psi)$
        \ENDIF
        \STATE Guess a letter $a\in\Sigma$\\
        \STATE $\pi' = \pi'\backslash a$ (using $\ResidueByLetter$)\label{line:NPResidue}\\
        \FOR{ each state $s\in S$ }\label{ln:NPmodelupdate}
            \STATE $Exp(s) = Exp(s)\backslash a$ (using $\ResidueByLetter$)\label{line:NPExpResidue}
        \ENDFOR
    \ENDFOR
\ENDIF
\IF{$\varphi = \hat{K_i}\psi$}  	
    \STATE $S' = \GetSetNP(\M,\psi)$
    \STATE Output $\{s\in S\mid \exists t\in S'\mbox{ and }t\sim_i s \}$
\ENDIF
\end{algorithmic}
\label{algo:GetSetNP}
\end{algorithm}

%----------------------------------

\begin{algorithm}[t]
%\color{red}
\caption{{$\ResidueByLetter$}}
\textbf{Input}: Regular expression $\pi$ and a letter $a\in\Sigma$\\
\textbf{Output}: Returns $\pi\backslash a$
\begin{algorithmic}[1]
\IF{$\pi\in \Sigma\cup\{\epsilon,\delta \}$}
    \IF{$\pi = a$}
        \STATE return $\epsilon$\;
    \ELSE
        \STATE return $\delta$\;
    \ENDIF
\ENDIF
\IF{$\pi = \pi_1 + \pi_2$}
    \STATE return \\
    \ \ \ \ \ $\ResidueByLetter(\pi_1, a) + \ResidueByLetter(\pi_2, a)$
\ENDIF
\IF{$\pi = \pi_1.\pi_2$}
    \STATE return $\ResidueByLetter(\pi_1, a).\pi_2$ \\ 
   \ \ \ \ \ \ \ \ \ \ \ \ \ \ \ \ \ \  \ \ \ \ \ \ $+\AuxOut(\pi_1).\ResidueByLetter(\pi_2,a)$
\ENDIF
\IF{$\pi = (\pi_1)^*$}
    \STATE return $\ResidueByLetter(\pi_1, a).(\pi_1)^*$
\ENDIF
\end{algorithmic}
\end{algorithm}
%-------------------------------------

\begin{algorithm}[tb]
%\color{red}
\caption{{$\AuxOut$}}
\textbf{Input}: A regular expression $\pi$\\
\textbf{Output}: Returns $\epsilon$ if $\epsilon\in\mathcal{L}(\pi)$ else $\delta$, where $\mathcal{L}(\delta)=\emptyset$\\
\begin{algorithmic}[1]
\STATE Create $A_\pi = <Q_\pi,\Sigma,\delta_\pi,q^\pi_0,F_\pi>$, the NFA for $\pi$\;
\IF{$q^\pi_0\in F_\pi$}
    \STATE return $\epsilon$
\ELSE
    \STATE return $\delta$\;
\ENDIF
\end{algorithmic}
\end{algorithm}

\ 

\

\begin{lemma}
Given a regular expression $\pi$ over $\Sigma$, and $a\in\Sigma$,  $\ResidueByLetter$ returns $\pi\regdiv a$ in polynomial time.
\label{lemma:residuebyletter}
\end{lemma}
\begin{proof}
The algorithm is a recursive one that directly follows the inductive definition~\ref{def:residue} of Residue. Hence the correctness and the complexity of the algorithm follows.
\end{proof}

\begin{lemma}\label{thm:NP}
Given a finite $\POL$ model $\M = \ldiamondarg{S,R,V,Exp}$, an $s\in S$ and a  $\starfree-\existential$ formula $\phi$, the algorithm $\mcNP$ is a polynomial time non-deterministic algorithm that outputs $\T$ iff $\M,s\vDash\phi$.

On other words, the the model-checking problem for the $\starfree-\existential$ fragment of $\POL$ is in $\NP$.
\end{lemma}

% \begin{restatable}{lemma}{NPCorrect}\label{thm:NP}
%     Given a finite $\POL$ model $\M = \ldiamondarg{S,R,V,Exp}$, an $s\in S$ and a $\POL$ formula of the $\starfree-\existential$ fragment $\phi$, $s\in\GetSetNP(\M, \phi)$ iff $\M,s\vDash\phi$. Also the algorithm $\GetSetNP$ is a non-deterministic polytime algorithm.
% \end{restatable}
\begin{proof}
    We prove the lemma in two parts: first we will prove the correctness of the algorithm $\GetSetNP$ - that is we show that given a finite $\POL$ model $\M = \ldiamondarg{S,R,V,Exp}$, an $s\in S$ and a  $\starfree-\existential$ formula $\phi$, $s\in\GetSetNP(\M, \phi)$ iff $\M,s\vDash\phi$. The correctness of the algorithm  $\mcNP$ follows immediately. 
    
    We then prove the complexity of the algorithm $\mcNP$.
    
    \ 
    
    \noindent\textbf{Proof of Correctness of $\GetSetNP$}  Let us start by proving the correctness of the algorithm $\GetSetNP$.
    This can be proved by induction over the size of $\phi$. 
    
    \textbf{Base Case.} Consider the case where $\phi = p$, where $p\in\BP$. In the IF case in $\ref{getsetnp:basecase}$, the set $S'$ is populated with all the worlds $s\in S$ where $p\in V(s)$. Hence, $\M,s\vDash\phi$ iff $s\in\GetSetNP(\M,\phi)$
    
    \textbf{Induction Hypothesis (IH).} Given a finite $\POL$ model $\M = \ldiamondarg{S,R,V,Exp}$, an $s\in S$ and a  $\starfree-\existential$ formula $\phi$, $s\in\GetSetNP(\M, \phi)$ iff $\M,s\vDash\phi$, where $|\phi|\leq k$, for an integer $k$.
    
    \textbf{Inductive Step.} 
    \begin{itemize}
    \item $\varphi = \neg{\psi}$
    \begin{align*}
        \M,s\vDash\neg{\psi}&\mbox{ iff } \M,s\nvDash\psi\hspace*{12.0cm}\\
        &\mbox{ iff }  s\notin \GetSetNP(\M,\psi)\mbox{, by IH}\\
        &\mbox{ iff }  s\in S\setminus \GetSetNP(\M,\psi)\\
        &\mbox{ iff }  s\in\GetSetNP(\M,\neg{\psi})
    \end{align*}
    \item $\varphi = \psi_1\vee\psi_2$
    \begin{align*}
         \M,s\vDash \psi_1\vee\psi_2 &\mbox{ iff } \M,s\vDash\psi_1\mbox{ or }\M,s\vDash\psi_2\hspace*{12.3cm}\\
        &\mbox{ iff }  s\in\GetSetNP(\M,\psi_1)\\
        &\ \ \ \ \ \ \ \mbox{ or }   s\in\GetSetNP(\M,\psi_2)\mbox{, by IH}\\
        &\mbox{ iff }  s\in\GetSetNP(\M,\psi_1) \\
        &\ \ \ \ \ \ \  \cup\GetSetNP(\M,\psi_2)\\
        &\mbox{ iff }  s\in\GetSetNP(\M,\psi_\vee\psi_2)
    \end{align*}
    \item $\varphi = \psi_1\wedge\psi_2$
    \begin{align*}
         \M,s\vDash \psi_1\wedge\psi_2 &\mbox{ iff }  \M,s\vDash\psi_1\mbox{ and }\M,s\vDash\psi_2 \hspace*{12.3cm}\\
        &\mbox{ iff }  s\in\GetSetNP(\M,\psi_1)\\
        &\ \ \ \ \ \ \  \mbox{and }s\in\GetSetNP(\M,\psi_2)\mbox{, by IH}\\
        &\mbox{ iff } s\in\GetSetNP(\M,\psi_1)\\
        &\ \ \ \ \ \ \ \cap\GetSetNP(\M,\psi_2)\\
        &\mbox{ iff } s\in\GetSetNP(\M,\psi_\wedge\psi_2)
    \end{align*}
    \item $\varphi = \hat{K_i}\psi$
    \begin{align*}
         \M,s\vDash\hat{K_i}\psi & \mbox{ iff } \exists t\sim_i s\mbox{ and }\M,t\vDash\psi\hspace*{13.0cm}\\ 
        &\mbox{ iff } \exists t\sim_i s\mbox{ and }t\in\GetSetNP(\M,\psi)\mbox{, by IH}\\
        &\mbox{ iff } \GetSetNP(\M,\hat{K_i}\psi)
    \end{align*}
    \item $\varphi = \ldiamondarg{\pi}\psi$
    \begin{align*}
        \M,s\vDash\ldiamondarg{\pi}\psi &\mbox{ iff } \exists w\in\LL(\pi): \LL(Exp(s)\regdiv w)\neq\emptyset\hspace*{13.0cm}\\ 
        &\ \ \ \ \ \ \ \mbox{ and }\M|_w,s\vDash\psi\\
        &\mbox{ iff } \exists w: |w|\leq |\pi|\mbox{ and }\LL(Exp(s)\regdiv w)\neq\emptyset\\ 
        &\ \ \ \ \ \ \ \mbox{ and }s\in\GetSetNP(\M|_w,\psi)\\
    \end{align*}
    Since $\pi$ is star-free, hence all words in $\LL(\pi)$ is size at most $|\pi|$. Loop in line~\ref{ln:NPguessloop} guesses $w\in\LL(\pi)$ letter by letter, residues $\pi$ (Line~\ref{line:NPResidue}) and updates model (the for loop starting from Line~\ref{ln:NPmodelupdate}), and recursively calls $\GetSetNP(\M|_w,\psi)$ once $\pi$ is exhausted (condition in Line~\ref{ln:piexhaust}), that is $w\in\LL(\pi)$. The residuation in Line~\ref{line:NPResidue} and the model updation can be done by $\ResidueByLetter$, which is correct by Lemma~\ref{lemma:residuebyletter}. Therefore, $\M,s\vDash\ldiamondarg{\pi}\psi$ iff $s\in\GetSetNP(\M,\ldiaarg{\pi}\psi)$
\end{itemize}
     
     \ 
     
    \noindent\textbf{Complexity of $\GetSetNP$}
    Now for the complexity of the algorithm, $\GetSetNP$  let us prove that the algorithm is a non-deterministic polytime algorithm. $\GetSetNP$ is a recursive algorithm that returns the worlds in $\M$ where $\phi$ holds. The algorithms labels each world $s$ in the $\M$ with $\psi\subseteq\phi$ iff $\psi$ holds in $s$.

    For each case of propositional operators, that is, $\phi = p\mid\ \ \neg{p}\mid\ \ \psi_1\vee\psi_2\mid\ \ \psi_1\wedge\psi_2$, assuming the worlds are labelled by the subformulas of $\phi$, deciding the worlds where $\phi$ holds require linear steps with respect to the worlds in $\M$.
    
    In case of $\phi = \hat{K_i}\psi$, for each world $s$ to be decided whether to label by $\phi$, at worst case at most all the related worlds of $s$ needs to be checked whether at least one of them is labelled by $\psi$. Hence, assuming all the worlds satisfying $\psi$ are labelled, this step takes quadratic steps with respect to the number of worlds in $\M$. 
    
    In the case of $\varphi = \ldiamondarg{\pi}\psi$, $\pi$ is promised to be star-free. Hence, any word $u\in\LL(\pi)$ is such that $|u|\leq|\pi|$. Hence, there exists a sequence of letters $w$ of length at most $|\pi|$, such that $\M|_w,s\vDash\psi$ if and only if $\M,s\vDash\ldiamondarg{\pi}\psi$. By Lemma~\ref{lemma:residuebyletter}, the residuation in Line~\ref{line:NPResidue} and Line~\ref{line:NPExpResidue} takes polynomial time. Also at each step, there are constant number of guesses (a letter from $\Sigma$).\\
    Also, there can be at most $|\phi|$ number of subformulas of $\phi$, each of size at most $|\phi|$.
    
    Hence $\GetSetNP$ is the $\NP$ algorithm for model checking problem where the input formula is promised to be $\starfree-\existential$. Thus the algorithm $\mcNP$ is also a non-deterministic polytime algorithm.

\end{proof}

\begin{lemma}\label{thm:NPH}
The model-checking problem for the $\starfree-\existential$ fragment of $\POL$ is $\NP$-hard.
\end{lemma}

%\starfreeexistNPC*
\begin{proof}
We shall prove this by a reduction from 3-SAT. % problem.
    %  The 3-SAT problem is defined as following:
    % Given a propositional formula $\varphi(x_1,\dots ,x_n)$ in 3-CNF form, having Boolean variables $\{ x_1,\dots ,x_n\}$ and $m$ clauses, does there exist an assignment of the variables such that $\varphi(x_1,x_2\dots ,x_n)$ is satisfiable?\\
    Given a 3-SAT CNF-formula $\varphi$ containing $n$ variables $\{x_1,\dots ,x_n \}$ and $m$ clauses $C = \{c_1,\dots ,c_m \}$, we shall construct our reduction. Consider the translations $\tr(x_i) = a_i$, $\tr(\neg{x_i}) = a'_i$ for all $i\in [n]$. Now we present the model checking instance that we construct from $\varphi$.
    \begin{itemize}
        \item  Alphabet: $\Sigma_\varphi = \cup_{i\in [n]}\{\tr(x_i), \tr(\neg{ x_i})\}$.
        %corresponding to positive and negative literal of the variable $x_i$ respectively.
        \item Model: $\M_\varphi = \langle S_\varphi,\sim_\varphi,V_\varphi,Exp_\varphi \rangle$, where
        \begin{itemize}
            \item $S_\varphi = \{1,\dots, m \}$.
            \item A single agent $1$ and $\forall i, j\in S_{\varphi}$, $i\sim_1 j$. %\todo{correct the indices: $i$ is for indexing the propositions, use $j$ instead}
            \item $V_\varphi(j) = \{p_j\}$ for all $j\in S_{\varphi}$. %\todo{use $j$ instead}
            \item For each clause $c_j\in C$, $Exp_\varphi(j) = (\sum_{l\in c_j}(\Sigma_{\varphi}\setminus\{\tr(\ell), \tr(\neg{\ell})\})^*\tr(\ell)(\Sigma_{\varphi}\setminus\{\tr(\ell), \tr(\neg{\ell})\})^*)$, where the sum is over the literals in the clause $c_j$. 
            % For example, for a clause $c_j = (x_3\vee \neg{x_4}\vee x_2)$ in $\varphi(x_1,..,x_n)$, $Exp_\varphi(j) = ((\Sigma\setminus\{a_3, a'_3\})^*a_3(\Sigma\setminus\{a_3,a'_3\})^* + (\Sigma\setminus\{a_4, a'_4\})^*a'_4(\Sigma\setminus\{a'_4, a_4\})^* + (\Sigma\setminus\{a_2,a'_2\})^*a_2(\Sigma\setminus\{a_2, a'_2\})^*)$.
            
            %\todo{add a remark that the star in Exp could also be removed no? simply repeat the pattern $n$ times (see the formula)}
        \end{itemize}
        \item Formula: $\psi := \ldiamondarg{a_1 + a'_1}\dots \ldiamondarg{a_n + a'_n}\bigwedge_{i\in [m]}(\hat{K_1}p_i)$.
        \item Starting world: Any $s$ from $S$.
    \end{itemize}
    All we need to show now is that the CNF-formula $\varphi$ is satisfiable iff $\M_\varphi,s\vDash \psi$.
%    \ldiamondarg{a_1 + a'_1}\dots \ldiamondarg{a_n + a'_n}\bigwedge_{i\in [m]}(\hat{K_1}p_i)$ . 
% The proof is similar as the hardness proof of Theorem \ref{theorem:starfree}
% and hence details are skipped. 

Let us start by assuming $\varphi$ is $\T$. We have to prove that $\M_\varphi,s\vDash\psi$.
Hence there exists an assignment $\sigma = (\ell_1,\ell_2,\ldots,\ell_n)$, where $\ell_i\in\{x_i,\neg{x_i}\}$, such that it evaluates $\varphi$ to $\T$. Now the consider the corresponding word $w = \tr(\ell_1)\tr(\ell_2)\ldots\tr(\ell_n)$. By construction, proving $\M_\varphi|_w,s\vDash\bigwedge_{i\in[m]}(\hat{K_1}p_i)$ proves our claim.

We can consider $\M_\varphi|_w$ to be updated from a series of updates $\M_\varphi|_{\tr(\ell_1)}, \M_\varphi|_{\tr(\ell_1)\tr(\ell_2)},\ldots, \M_\varphi|_{\tr(\ell_1)\tr(\ell_2)\ldots\tr(\ell_n)}$. Let us consider a world $j$ in the model $\M_\varphi$, corresponding to the clause $c_j = (\ell_p\vee \ell_q\vee \ell_r)$, where $\ell_p, \ell_q, \ell_r$ are literals in the clause. Note that, if neither of the literals in $c_j$ contains $x_i$, the update of the model with $\tr(x_i)$ or $\tr(\neg{x_i})$ does not affect the world $j$. Now, on the other hand, without loss of generality, since $\sigma$ is a satisfying assignment, consider literal $\ell_p$ is in $\sigma$, since at least one of the literals in $c_j$ has to be in the satisfying assignment. Note that after  updating corresponding to $\tr(\ell_p)$, the $Exp(j)$ will have the residue regular expression $(\Sigma\setminus\{\tr(\ell_p), \tr(\neg{\ell_p})\})^*$ which guarantees the survival of the world $j$ in future updates.

Hence, since $\sigma$ is a satisfying assignment, every clause will have at least one literal in $\sigma$, hence guaranteeing the survival of all the worlds in $\M_\varphi|_w$. Since all the world survives and $s\sim_1 t$, for every $t\in S_\varphi$, hence $\M_\varphi|_w,s\vDash\bigwedge_{j\in[m]}\hat{K_j}p_j$.

Conversely, let $\varphi$ is unsatisfiable. Hence, for any assignment $\sigma = (\ell_1, \ell_2,\ldots, \ell_n)$, there exists at least one clause, say $c_j = (\ell_p\vee \ell_q\vee \ell_r)$, such that $\neg{\ell_p}, \neg{\ell_q}, \neg{\ell_r}$ is in the $\sigma$. Note update in the world $j$ corresponding to clause $c_j$. After the update corresponding to $\tr(\neg{\ell_p})$, the term $(\Sigma\setminus\{\tr(\ell_p),\tr(\neg{\ell_p})\})^*\tr(\ell_p)(\Sigma\setminus\{\tr(\ell_p),\tr(\neg{\ell_p})\})^*$ in $Exp(j)$ becomes $\delta$ (regular expression for the empty regular language). Same occurs for the update corresponding to $\tr(\neg{\ell_q})$ and $\tr(\neg{\ell_r})$. Hence the $Exp(j)$, after all the three updates, becomes $\delta$, due to which the world $j$ does not survive. Hence $\M_\varphi|_w,s\nvDash\bigwedge_{j\in[m]}\hat{K_j}p_j$, where $w = \tr(\ell_1)\tr(\ell_2)\ldots\tr(\ell_n)$. 

% For every satisfying assignment, it can be proved that every world survives in the projected model and also if the formula is unsatisfiable, for every assignment at least one world does not survive (corresponding to the $\Fa$ clause).
\end{proof}

\begin{algorithm}[h]
%\color{red}
\caption{$\mcWORDS$}
\textbf{Input}: $\M = \ldiamondarg{S, R, V, Exp}$, $s\in S$, $\varphi$, where $\varphi$ is a $\word$ Formula.\\
\textbf{Output}: Returns $\T$ iff $\M,s\vDash\varphi$
\begin{algorithmic}[1] %[1] enables line numbers
\IF{$s\in \GetSet(\M, \phi)$}
\STATE Return $\T$
\ENDIF
%\ELSE{Return $\F$}
\end{algorithmic}
\label{algo:mcWord}
\end{algorithm}
\vspace{-0.5cm}
\begin{algorithm}[h]
\caption{$\GetSet$}
\textbf{Input}: $\M = \ldiamondarg{S, R, V, Exp}$, $\varphi$, where $\varphi$ is a $\word$ Formula.\\
\textbf{Output}: Returns set of states $S'\subseteq S$ such that $\M,s\vDash\varphi$ for all $s\in S'$
\begin{algorithmic}[1] %[1] enables line numbers
\IF{$\varphi=p$ is a propositional variable }\label{getset:basecase}
    \STATE $S'=\emptyset$\;
    \FOR{$s\in S$}
        \IF{$p\in V(s)$}
            \STATE $S' = S'\cup \{ s\}$\;
        \ENDIF
    \ENDFOR
    \STATE Output $S'$\;
\ENDIF
\IF{$\varphi = \neg\psi$}
	\STATE Output $S\setminus \GetSet(\M,\psi)$
\ENDIF
\IF{$\varphi = \psi_1 \vee \psi_2$}
    \STATE Output $\GetSet(\M,\psi_1)\cup\GetSet(\M,\psi_2)$
\ENDIF
\IF{$\varphi = \ldiamondarg{w}\psi$}
    \STATE $\M' = <S',R',V',Exp'> = \M$\;
    \FOR{$s\in S'$}\label{getset:residuemodelstart}
        \STATE $Exp'(s) = Exp(s)\setminus w$\label{getset:expresidue}
        \IF{$Exp'(s) = \delta$}
            \STATE $S' = S'\setminus \{s\}$
            \STATE $R' = R'\setminus \{\{s,t \}\}$ for every $t\in S$\;
        \ENDIF
    \ENDFOR\label{getset:residuemodelfin}
    Output $\GetSet(\M',\psi)$
\ENDIF
\IF{$\varphi = \hat{K_i}\psi$}  	
    \STATE $S' = \GetSet(\M,\psi)$
    \STATE Output $\{s\in S\mid \exists t\in S'\mbox{ and }t\sim_i s \}$
\ENDIF
\end{algorithmic}
\end{algorithm}

\subsection{Model-checking for the $\word$ fragment of $\POL$}

In this section we present the proof of the following theorem. 

\wordP*

 \begin{proof}
The model checking algorithm for $\POL$, when $\pi$'s are words, can be designed in a similar way as the folklore recursive model checking algorithm for epistemic logic. Only modification is when, checking whether $\M, s\vDash \langle\pi\rangle \phi$ recursively call $\M|_w, s\vDash \phi$ where $\LL(\pi) = \{w\}$.

We use the algorithm $\mcWORDS$ for model checking the $\word$ fragment. 
The algorithm $\mcWORDS$ calls the subroutine $\GetSet$ which is a polytime algorithm that takes a model $\M$ and a word-formula $\phi$ and outputs the set of all states $s$ such that $\M,s\vDash \phi$. The correctness of the 
algorithm $\mcWORDS$ follows from the correctness of the algorithm $\GetSet$.
The proof of correctness of the algorithm $\GetSet$ is presented in Lemma~\ref{lem:P}.
% As the $\pi$'s are words, there are a linear number of updated models with respect to $|w|$ to be considered. Each world in the model $\M$ is recursively labeled by either $\varphi_i$ or $\neg{\varphi_i}$, where $\varphi_i$ is a subformula of $\varphi$. Since the recursion is on the size of $\varphi$, $1\leq i\leq |\varphi|$. The proof of correctness of the algorithm $\GetSet$ is presented in Lemma~\ref{lem:P}.
\end{proof}

\begin{restatable}{lemma}{wordCorrect}\label{lem:P}
    Given a finite $\POL$ model $\M = \ldiamondarg{S,R,V,Exp}$, an $s\in S$ and a $\word$ formula $\phi$, $s\in\GetSet(\M, \phi)$ iff $\M,s\vDash\phi$.
\end{restatable}
\begin{proof}
    This can be proved by induction over the size of $\phi$. 
    
    \textbf{Base Case.} Consider the case where $\phi = p$, where $p\in\BP$. In the IF case in $\ref{getset:basecase}$, the set $S'$ is populated with all the worlds $s\in S$ where $p\in V(s)$. Hence, $\M,s\vDash\phi$ iff $s\in\GetSet(\M,\phi)$.
    
    \textbf{Induction Hypothesis.} Given a finite $\POL$ model $\M = \ldiamondarg{S,R,V,Exp}$, an $s\in S$ and a $\POL$ formula of the $\word$ fragment $\phi$, $s\in\GetSet(\M, \phi)$ iff $\M,s\vDash\phi$, where $|\phi|\leq k$, for an integer $k$.
    
    \textbf{Inductive Step.} 
    \begin{itemize}
    \item $\varphi = \neg{\psi}$
    \begin{align*}
        \M,s\vDash\neg{\psi}  &\mbox{ iff } \M,s\nvDash\psi\hspace*{13.0cm}\\ 
        &\mbox{ iff } s\notin \GetSet(\M,\psi)\mbox{, by IH}\\
        &\mbox{ iff } s\in S\setminus \GetSet(\M,\psi)\\
        &\mbox{ iff } s\in\GetSet(\M,\neg{\psi})
    \end{align*}
    \item $\varphi = \psi_1\vee\psi_2$
    \begin{align*}
        \M,s\vDash \psi_1\vee\psi_2 &\mbox{ iff } \M,s\vDash\psi_1\mbox{ or }\M,s\vDash\psi_2\hspace*{12.3cm}\\
        &\mbox{ iff } s\in\GetSet(\M,\psi_1)\\
        &\ \ \ \ \ \ \ \mbox{ or }s\in\GetSet(\M,\psi_2)\mbox{, by IH}\\
        &\mbox{ iff } s\in\GetSet(\M,\psi_1)\cup\GetSet(\M,\psi_2)\\
        &\mbox{ iff } s\in\GetSet(\M,\psi_\vee\psi_2)
    \end{align*}
    \item $\varphi = \hat{K_i}\psi$
    \begin{align*}
        \M,s\vDash\hat{K_i}\psi &\mbox{ iff } \exists t\sim_i s\mbox{ and }\M,t\vDash\psi\hspace*{13.0cm}\\ 
        &\mbox{ iff } \exists t\sim_i s\mbox{ and }t\in\GetSet(\M,\psi)\mbox{, by IH}\\
        &\mbox{ iff } \GetSet(\M,\hat{K_i}\psi)
    \end{align*}
    \item $\varphi = \ldiaarg{w}\psi$\\
    Since $w$ is a word, $Exp(s)\backslash w$ is calculated in line \ref{getset:expresidue}, and hence for all $s\in S$ in the loop in \ref{getset:residuemodelstart}-\ref{getset:residuemodelfin}. Also a certain state $t(\in S)\notin S'$ iff $Exp(t)\backslash w = \delta$, that is, $\mathcal{L}(Exp(t)\backslash w) = \emptyset$. By Lemma~\ref{lemma:residuebyletter}, we have the correctness of the residuation. Hence after the termination of loop \ref{getset:residuemodelstart}-\ref{getset:residuemodelfin}, $\M' = \M|_w$.
    \begin{align*}
        \M,s\vDash\ldiamondarg{w}\psi &\mbox{ iff } \mathcal{L}(Exp(s)\backslash w)\neq\emptyset\hspace*{12cm}\\ 
        &\ \ \ \ \ \ \mbox{ and }\M|_w,s\vDash\psi\mbox{, since $w$ is word}\\
        &\mbox{ iff } \mathcal{L}(Exp(s)\backslash w)\neq\emptyset\\
        &\ \ \ \ \ \ \mbox{ and }s\in\GetSet(\M',\psi)\mbox{, by IH}\\
        &\mbox{ iff } s\in\GetSet(\M,\ldiamondarg{w}\psi)
    \end{align*}
\end{itemize}
\end{proof}

% \section{Efficient  model-checking algorithm for the $\starfree-\existential$ fragment of  $\POL$}
% %Reduction of the existential star-free fragment POL model checking to SAT}
% \label{appendix:starfreeexistentialToSAT}
 %François has written that

\end{document}